\documentclass[journal, onecolumn, draftcls]{IEEEtran}

\usepackage{amsthm,amsfonts,amssymb,amsmath, microtype}
\usepackage[all]{xy}
\usepackage{xr-hyper}
\usepackage[colorlinks=true]{hyperref}
\usepackage{verbatim}
\usepackage{pdfsync}
\usepackage{xcolor}
\usepackage{graphicx}
\usepackage{todonotes}
\usepackage{subfig}
\usepackage{dblfloatfix}

\newtheorem{theorem}{Theorem}

\newtheorem{lemma}[theorem]{Lemma}

\newtheorem{corollary}[theorem]{Corollary}
\newtheorem{definition}[theorem]{Definition}
\newtheorem{remark}[theorem]{Remark}

\newcommand{\etal}[1]{et al.}
\newcommand{\vc}[1]{\mathbf{#1}}
\newcommand{\mt}[1]{\mathbf{#1}}
\newcommand{\set}[1]{\mathcal{#1}}
\newcommand{\diag}{\text{diag}}

\newcommand{\supp}{\text{supp}}
\newcommand{\Order}{\mathcal{O}}

\newcommand{\setM}{\set{M}^+}

\begin{document}

\title{Robust Nonnegative Sparse Recovery and the Nullspace Property
  of 0/1 Measurements\thanks{
    This paper has been presented in part at the 
    2016 IEEE Information Theory Workshop - ITW 2016, Cambridge, UK.
  }}
\author{%
  \IEEEauthorblockN{Richard Kueng\IEEEauthorrefmark{1},
    Peter Jung\IEEEauthorrefmark{2}}\\
  \IEEEauthorblockA{\IEEEauthorrefmark{1}\IEEEauthorblockA{Institute
      for Theoretical Physics, University of Cologne}}\\
  \IEEEauthorblockA{\IEEEauthorrefmark{2}\IEEEauthorblockA{Communications and Information Theory Group, Technische Universit\"{a}t Berlin}}\\
  {\em rkueng@thp.uni-koeln.de, peter.jung@tu-berlin.de}
  \vspace*{-1cm}
}

\maketitle


\begin{abstract}
  We investigate recovery of nonnegative vectors
  from non-adaptive compressive measurements in the presence of noise of unknown
  power. In the absence of noise, existing results in the literature identify properties of the measurement
 that assure uniqueness in the  non-negative orthant.
  By linking such uniqueness results to nullspace
  properties, we deduce uniform and robust compressed sensing guarantees for
  nonnegative least squares. No $\ell_1$-regularization is required. 
  As an important proof of principle, we establish that $m\times n$ random
  i.i.d.~$0/1$-valued Bernoulli matrices obey the required conditions
  with overwhelming
  probability provided that $m=\Order(s\log(n/s))$.   
  We achieve this by establishing the robust nullspace property
  for random $0/1$-matrices---a novel result in its own right.  
  Our analysis is motivated by applications in 
  wireless network activity detection. 
\end{abstract}


\section{Introduction}
Recovery of lower complexity objects by observations far below the
Nyquist rate has applications in physics, applied math, and many
engineering disciplines. Moreover, it is one of the key tools for facing challenges in data processing (like big data and the
Internet of Things), 
wireless communications (the 5th generation of the mobile cellular network) and large scale network control. 
Compressed Sensing (CS), with its original goal of
recovering sparse or compressible vectors, has, in particular, stimulated the
research community to investigate further in this direction. The aim
is to identify compressibility and low-dimensional structures  which 
allow the recovery from low-rate samples with efficient
algorithms. In many applications, the objects of interest exhibit further structural constraints which should be exploited in reconstruction algorithms. 
Take, for instance, the following setting which appears naturally in communication protocols:
The components of sparse information carrying vectors are
taken from a finite alphabet, or the data vectors are lying in specific
subspaces. 
Similarly, in network traffic estimation and anomaly detection from end-to-end
measurements, the parameters are restricted to particular
low-dimensional domains.
Finally, the signals occurring in imaging problems are typically constrained to non-negative intensities.

Our work is partially inspired by the 
task of identifying sparse network activation patterns in a large-scale 
asynchronous wireless network: Suppose that, in order to indicate its presence, each active 
device node transmits an individual sequence into a noisy wireless channel. All such sequences are multiplied with
individual, but unknown, channel amplitudes\footnote{%
  This can be justified under certain assumptions like
  pre-multiplications using channel reciprocity in time-division multiplexing.}
and finally superimpose
at the receiver. The receiver's task then is to detect all active devices and the corresponding channel amplitudes from this global superposition (note that each device is uniquely characterized by the sequence it transmits).
This problem can be re-cast as the task of estimating non-negative sparse vectors from noisy linear observations. 

Such non-negative and sparse structures also arise naturally in certain empirical
inference problems, like network tomography \cite{Vardi:1996,Castro:2004},
statistical tracking (see e.g.\ \cite{Boyd2003})
and compressed imaging of intensity patterns \cite{Donoho92}.
The underlying mathematical problem has received considerable
attention in its own right \cite{Fuchs2005,Zhang:2010,Khajehnejad2011,Meinshausen2013,Slawski2013,Foucart:NNLS:2014}.
It has been shown that measurement matrices
$\mt{A}\in\mathbb{R}^{m\times n}$ coming from \emph{outwardly $s$-neighborly
polytopes} \cite{Donoho2005} and matrices $\mt{A}$ whose
\emph{row span intersects the positive orthant}\footnote{See Eq.~\eqref{eq:def:Mset} below for a precise definition.}
\cite{Bruckstein2008}
maintain an intrinsic uniqueness property for non-negative, $s$-sparse
vectors. These carry over to the under-determined setting ($m<n$).
Such uniqueness properties in turn allow for entirely avoiding CS algorithms in the reconstruction step.
From an algorithmic point of view, this is highly beneficial. 
However, all the statements mentioned above focus on idealized scenarios, where no noise is present in the sampling procedure. 

Motivated by device detection, we shall overcome this 
idealization and devise non-negative recovery protocols that are robust towards 
any form of additive noise.
Our results have the added benefit that no a-priori bound on the 
noise step is required in the algorithmic reconstruction. 

\subsection{Main Results}

Mathematically, we are interested in recovering sparse, entry-wise nonnegative vectors $\vc{x} \geq \vc{0}$ in $\mathbb{R}^n$ from $m \ll n$ noisy linear measurements of the form
$
y_i = \vc{a}_i^T \vc{x} + e_i
$.
Here, the vectors $\vc{a}_i \in \mathbb{R}^n$ model the different
linear measurement operations and $e_i$ is additive noise of arbitrary size and nature. 
By encompassing all $\vc{a}_i$'s as rows of a sampling matrix
$\mt{A}\in\mathbb{R}^{m\times n}$ and defining $\vc{y} = (y_1,\ldots,y_m)^T$, as well as $\vc{e} = (e_1,\ldots,e_m)^T$, such a sampling procedure can succinctly be written as 
\begin{equation}
\vc{y} = \mt{A} \vc{x} + \vc{e}. \label{eq:measurements}
\end{equation}
Several conditions on $\mt{A}$ are known to be sufficient to ensure that
a sparse vector $\vc{x}$  can be robustly estimated from measurements $\vc{y}$. 
Here, we focus on \emph{uniform} reconstruction guarantees. These assure recovery of all $s$-sparse vectors simultaneously.
While several sufficient criteria for uniform recovery exist, the \emph{nullspace property} (NSP) is both necessary and sufficient.
In order to properly define a robust version of the NSP, see e.g.\ \cite[Def.~4.21]{Foucart2013}, we need to introduce some notation:
Fix $\vc{x} \in \mathbb{R}^n$ and let $S \subset [n] = \left\{1,\ldots,n \right\}$ be a set. We denote the restriction of $\vc{x}$ to $S$ by $\vc{x}_S$ (i.e. $\left( \vc{x}_S \right)_i = \vc{x}_i$ for $i \in S$ and $\left( \vc{x}_S \right)_i = 0$ else). Let $\bar{S}$ be the complement of $S$ in $[n]$, such that $\vc{x}=\vc{x}_S + \vc{x}_{\bar{S}}$.

\begin{definition}[$\ell_2$-robust nullspace property] \label{def:nsp}
  A $m \times n$ matrix $\mt{A}$ satisfies the \emph{$\ell_2$-robust null
    space property} of order $s$ with parameters $\rho \in
  (0,1)$ and $\tau >0$, if:
  \begin{equation*}
    \| \vc{v}_S \|_{\ell_2} \leq \frac{\rho}{\sqrt{s}} \|
    \vc{v}_{\bar{S}} \|_{\ell_1} + \tau \left\| \mt{A} \vc{v} \right\|_{\ell_2}
    \quad \forall \vc{v}\in\mathbb{R}^n
  \end{equation*}
  holds for all $S\subset[n]$ with $|S|\leq s$.
\end{definition}

This property implies that no $s$-sparse vectors lie in the kernel (or nullspace) of $\mt{A}$.
Importantly, validity of the NSP also implies
\begin{equation}
\| \vc{x} - \vc{z} \|_{\ell_2} \leq \frac{C}{\sqrt{s}} \left( \| \vc{z} \|_{\ell_1} - \| \vc{x} \|_{\ell_1}\right) + D \tau \| \mt{A}(\vc{x}-\vc{z}) \|_{\ell_2}, \label{eq:nsp_implication}
\end{equation}
for any $s$-sparse $\vc{x} \in \mathbb{R}^n$ and every $\vc{z} \in \mathbb{R}^n$ \cite[Theorem~4.25]{Foucart2013}. The constants $C,D$ only depend on the NSP parameter $\rho$ and we refer to Formula~\eqref{eq:nsp_alternative} below for explicit dependencies.
In turn, this relation implies that every $s$-sparse vector $\vc{x}$ can be reconstructed from noisy measurements of the form \eqref{eq:measurements} via \emph{basis pursuit denoising} (BPDN):
\begin{equation}
  \vc{x}^\sharp_\eta=\arg\min\lVert\vc{z}\rVert_{\ell_1}\quad\text{s.t.}\quad
  \lVert \mt{A}\vc{z}-\vc{y}\rVert_{\ell_2}\leq\eta.
  \label{eq:bpdn}
\end{equation}
Here, $\eta$ must be an a-priori known upper bound on the noise strength in \eqref{eq:measurements}: $\eta \geq \| \vc{e} \|_{\ell_2}$. 
Our first main technical contribution is a substantial strengthening of Formula~\eqref{eq:nsp_implication} that is valid for non-negative $s$-sparse vectors ($\vc{x} \geq \vc{0}$):

\begin{theorem} \label{thm:main_theorem}
  Suppose that $\mt{A}$ obeys the NSP of order $s \leq n$ and moreover
  admits a strictly-positive linear combination of its rows: $\exists \vc{t} \in \mathbb{R}^m$ such that $\vc{w}=\mt{A}^T \vc{t} >\vc{0}$.
  Then, the following bound holds for any $s$-sparse  $\vc{x} \geq \vc{0}$ and any $\vc{z} \geq \vc{0}$:
  \begin{equation}
    \| \vc{x} - \vc{z} \|_{\ell_2} \leq  D' \left( \|\vc{t}\|_{\ell_2}+\tau) \right)\left\| \mt{A} (\vc{z}-\vc{x} ) \right\|_{\ell_2}. \label{eq:main_bound}
  \end{equation}
  The constant $D'$ only depends on the quality of NSP and the conditioning of the strictly positive vector $\vc{w}$.
\end{theorem}

This statement is a simplified version of Theorem~\ref{thm:main} below and we refer to this statement for a more explicit presentation.
The crucial difference between \eqref{eq:main_bound} and \eqref{eq:nsp_implication} is the fact that no $(\|\vc{z} \|_{\ell_1} - \| \vc{x} \|_{\ell_1} )$-term occurs in the former. This term is responsible for the $\ell_1$-regularization in BPDN. Theorem~\ref{thm:main_theorem} highlights that this is not necessary in the non-negative case. Instead, a simple nonnegative least squares regression suffices:
\begin{equation}
  \vc{x}^\sharp = \underset{\vc{z} \geq \vc{0}}{\arg \min} \left\| \mt{A}\vc{z}-\vc{y} \right\|_{\ell_2}. \label{eq:least_squares}
\end{equation}
Under the pre-requisites of Theorem~\ref{thm:main_theorem}, the solution of this optimization problem stably reconstructs any non-negative $s$-sparse vector from noisy measurements \eqref{eq:measurements}. We refer to Sec.~\ref{sub:nsp_nn} for a derivation of this claim.
Here, we content ourselves with pointing out that this
recovery guarantee is (up to multiplicative constants) as strong as
existing ones for different reconstruction algorithms. These include the LASSO and Dantzig selectors, as well as basis pursuit denoising (BPDN)
(see \cite{Foucart2013} and references therein). 
However, on the contrary to them, algorithms for solving
\eqref{eq:least_squares} require neither an 
explicit a-priori bound $\eta \geq \| \vc{e} \|_{\ell_2}$ on the
noise, nor an $\| \cdot \|_{\ell_1}$  regression term. This {\em
  simplicity} is caused by the non-negativity constraint $\vc{z} \geq \vc{0}$ and the geometric restrictions it imposes.
Also, these assertions stably remain true if we consider approximately sparse target vectors instead of perfectly sparse ones (see Theorem~\ref{thm:main} below).

In order to underline the applicability of
Theorem~\ref{thm:main_theorem}, we consider nonnegative $0/1$-Bernoulli sampling matrices and prove that they meet the requirements of said statement with high probability (w.h.p).
This in turn implies:

\begin{theorem} \label{thm:main2}
Let $\mt{A}$ be a sampling matrix whose entries are independently chosen from a $0/1$-Bernoulli distribution with parameter $p \in [0,1]$, i.e. $\mathrm{Pr}[1] = p$ and $\mathrm{Pr}[0]=1-p$. 
Fix $s \leq n$ and set
\begin{equation}
m \geq C \alpha (p) s \left( \log \left(\frac{\mathrm{e}n}{s}\right) + \beta (p) \right)
\end{equation}
where $\alpha(p),\beta(p)$ are constants depending only on $p$. 
Then, with probability at least $1- (n+1) \mathrm{e}^{-C'p^2(1-p)^2m }$, $\vc{A}$ allows for stably reconstructing \emph{any} non-negative $s$-sparse vector $\vc{x}$ from $\vc{y}=\mt{A}+\vc{e}$ via \eqref{eq:least_squares}. The solution $\vc{x}^\sharp$ of \eqref{eq:least_squares} is guaranteed to obey
\begin{equation*}
\| \vc{x}^\sharp - \vc{x}\|_{\ell_2} \leq \frac{E'}{\sqrt{p(1-p)}^3} \frac{ \| \vc{e}\|_{\ell_2}}{\sqrt{m}},
\end{equation*}
where $E'$ is constant.
\end{theorem}
We emphasize two important aspects of this result:
\begin{enumerate}
\item $0/1$-Bernoulli
  matrices obey the NSP with overwhelming probability. This novel statement alone assures robust
  sparse recovery via BPDN \eqref{eq:bpdn}. Moreover, the required sampling rate is proportional to $s \log (n/s)$ which is optimal. 
\item For non-negative vectors we overcome traditional $\ell_1$-regularization. We demonstrate this numerically in Figure~\ref{fig:nnls:noiseless}.
\end{enumerate}
Up to our knowledge, this is the first rigorous proof that
$0/1$-matrices tend to obey a strong version of the nullspace
property. The main difference to most existing NSP and RIP results is the
fact that the individual random entries of $\mt{A}$ are not centered, ($\mathbb{E} \left[ \mt{A}_{k,j} \right] = p \neq 0$). 
Thus, the covariance matrix of $\mt{A}$ admits a condition number
of $\kappa(\mathbb{E}[\mt{A}^T\mt{A}]) = 1 + \frac{pn}{1-p}$, which underlines the ensemble's anisotropy.
Traditional proof techniques, like establishing an RIP, are either not applicable in such a setting, or yield sub-optimal results \cite{RudelsonZhou13, KuengGross14}. 
This is not true for Mendelson's small ball method
\cite{Mendelson15,KoltchinskiMendelson15} (see also \cite{Tropp2015}),
which we employ in our proof of Theorem~\ref{thm:main2}.
We refer to \cite{Dirksen16} for an excellent survey about the applicability of Mendelson's small ball method in compressed sensing.
In the conceptually similar problem of reconstructing low rank matrices
from rank-one projective measurements 
(which arises e.g.\ from the PhaseLift approach for phase retrieval
\cite{Candes13}), applying this technique allowed for establishing
strong null space properties, despite a similar degree of anisotropy.

Finally, we point out that the constant $\alpha (p)$ in Theorem~\ref{thm:main2} diverges for $p \to 0,1$. This is to be expected, because the inverse problem becomes ill-posed in this regime of sparse (or co-sparse) measurements. 
Despite our efforts, we do not expect $\alpha (p)$ to be tight in this interesting parameter regime and leave a more detailed analysis of this additional parameter dependence for future work, see Remark~\ref{rem:parameter} below.

\begin{figure}
  \begin{center}
    \includegraphics[width=.6\linewidth]{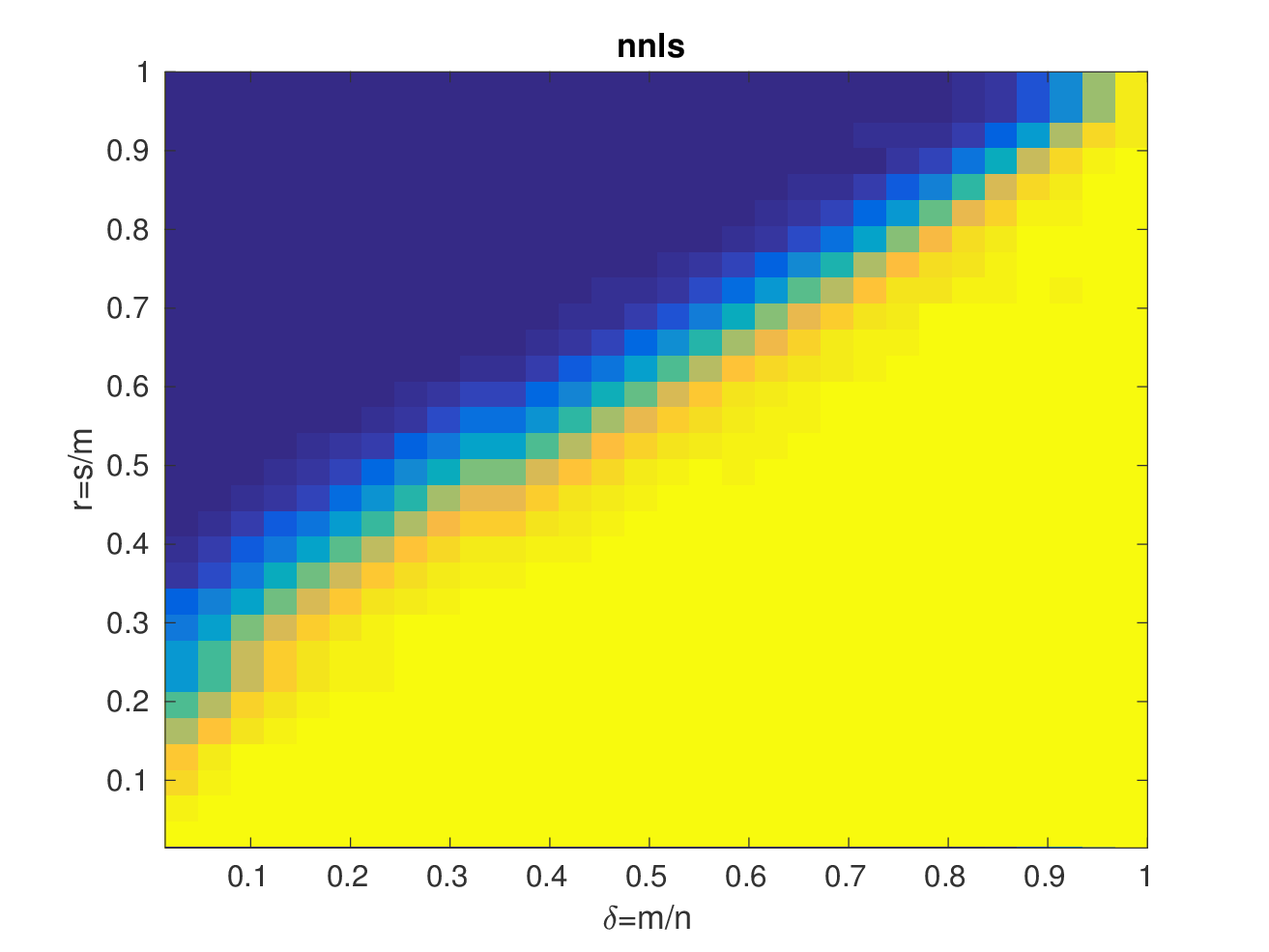} 
  \end{center}
  \caption{Phase transition for NNLS in \eqref{eq:least_squares} for 
    i.i.d.\ $0/1$-Bernoulli measurement matrices in the noiseless case. 
    More details are given in Section \ref{sec:numeval}.
} \label{fig:nnls:noiseless}
\end{figure}

{\em Organization of the Paper:}
In Section \ref{sec:sysmodel} we 
explain our motivating application in more detail and rephrase activity detection as a 
nonnegative sparse recovery problem. Then, we provide an overview 
on prior work and known results regarding this topic. 
In Section \ref{sec:nsp} we show that recovery guarantees in the presence
of noise are governed by the \emph{robust nullspace property} 
(see here \cite{Foucart2013}) \emph{under nonnegative constraints}.
Finally, in Section \ref{sec:bin:model} we analyze
binary measurement matrices having i.i.d.\ random $0/1$-valued entries.
We prove that such matrices admit
the NSP with overwhelming probability and moreover meet the additional requirement of Theorem~\ref{thm:main_theorem}.

\section{System Model and Problem Statement}
\label{sec:sysmodel}

\subsection{Activity Detection in Wireless Networks}
Let $\mt{A}=(\vc{s}_1 | \cdots | \vc{s}_n )\in\mathbb{R}^{m\times n}$ be a matrix
with $n$ real columns $\vc{s}_j\in\mathbb{R}^m$.
In our network application \cite{Yunyan:icassp16}, the columns $\vc{s}_j$ are the individual sequences
of length $m$ transmitted by the active devices. 
These sequences are transmitted simultaneously and each of them is
multiplied by an individual amplitude that 
depends on transmit power and other channel conditions. In practice
such a scenario can be achieved by using the channel reciprocity
principle in time-division multiplexing. This assures that the devices have
knowledge about the complex channel coefficients and may perform a
pre-multiplication to correct for the phase.
All these modulated sequences are superimposed at a single receiver,
because the wireless medium is shared by all devices.
We model such a situation by an unknown non-negative vector
$\vc{0}\le \vc{x}\in\mathbb{R}^n$, where $x_i>0$ indicates that a device with
sequence $i$ is active with amplitude $x_i$ ($x_i=0$ implies that a device is inactive).
We point out that, due to path loss in the channel, the individual received
amplitudes $x_i$ of each active device are unknown to the receiver as well. 
Here, we focus on networks that contain a large number $n$ of registered devices, 
but, at any time, only a small unknown fraction, say $s \ll n$, of these devices are active.

Communicating activity patterns, that is $\supp(\vc{x})=\{i\,:\,x_i\neq
0\}$, and the corresponding list of received amplitudes/powers ($\vc{x}\geq \vc{0}$ itself)
in a traditional way would require $\Order(n)$ resources.
Here, we aim for a reduction of the signaling time $m$ by exploiting the facts that
(i) $\vc{x}\geq \vc{0}$ is non-negative and (ii) the vector $\vc{x}$ is $s$-sparse, i.e.\ 
$\lVert \vc{x}\rVert_{\ell_0}\leq s$. Hence, we focus on the regime
$s\leq m\ll n$. Obviously, in such a scenario the resulting system of linear equations cannot be
directly inverted. A reasonable approach towards recovery is
to consider the program:
\begin{equation*}
  \arg\min\lVert\vc{z}\rVert_{\ell_0}\,\,\,\text{s.t.}\,\,\,
  \mt{A}\vc{z}=\vc{y}\,\,\&\,\, \vc{z}\geq \vc{0}
\end{equation*}
Combinatorial problems of this type are infamous for being NP-hard in general.
A common approach to circumvent this obstacle is to consider convex relaxations.
A prominent relaxation is to replace
$\lVert\cdot\rVert_{\ell_0}$ with the $\ell_1$-norm. The resulting algorithm can then be re-cast as an efficiently solvable linear program.
However, such approaches become more challenging when robustness
towards additive noise is required. In particular, if the type and the strength of
the noise is itself unknown.
In our application, noisy contributions inevitably arise due to quantization,
thermal noise and other interferences.
If the noisy measurements are of the form \eqref{eq:measurements} (i.e.
$\vc{y}=\mt{A}\vc{x}+\vc{e}$, where the vector $\vc{e}$ is an additive
distortion) a well-known modification is to consider the BPDN
\eqref{eq:bpdn} but with an additional nonnegativity constraint:
\begin{equation}
  \arg\min\lVert\vc{z}\rVert_{\ell_1}\,\,\,\text{s.t.}\,\,\,
  \lVert \mt{A}\vc{z}-\vc{y}\rVert_{\ell_2}\leq\eta\,\,\&\,\, \vc{z}\geq \vc{0}.
  \label{eq:bpdn:nn}
\end{equation}
While this problem is algorithmically a bit more complicated than \eqref{eq:bpdn}, it is still convex and
computationally tractable (in principle). In practice, further modifications are
necessary to solve such problems sufficiently fast and
efficiently, see \cite{Vila2014,Yunyan:icassp16}.
However, having access to an a-priori bound $\eta$ on $\lVert \vc{e}
\rVert_{\ell_2}$ is essential for (i) posing 
this problem and (ii) solving it using certain algorithms that involve stopping criteria, or other conditions that depend on the noise level.
Suppose, for instance, that $\vc{e}$ is i.i.d.\ normal distributed. Then $\lVert \vc{e} \rVert^2_{\ell_2}$
admits a $\chi^2$-distribution of order $m$ and  feasibility is assured w.h.p., when taking $\eta$ in terms of second moments.
However, much less is known for different noise distributions. This in particular includes situations where second moment information about the noise is challenging to acquire.

One option to tackle problems of this kind is to establish a \emph{quotient property} for the measurement matrix $\mt{A}$ \cite{Foucart2013}. 
However, this property is geared towards Gaussian measurements and 
is challenging to establish for different random models of $\mt{A}$. 
In this paper we show that another condition---namely that $\vc{A}$ admit a strictly positive linear combination of rows---allows for drawing similar conclusions.

\subsection{Prior Work on Recovery of Nonnegative Sparse Vectors}
One of the first works  on non-negative compressed sensing is due to
Donoho \etal in \cite{Donoho92} on the ``nearly black object''. It furthers the
understanding of the ``maximum entropy inversion'' method to recover sparse
(nearly-black) images in radio astronomy. 
Donoho and Tanner investigated this
subject more directly in Ref.~\cite{Donoho2005}. 
The central question is: what properties of $\mt{A}$ intrinsically ensure that only one solution is feasible for any $s$-sparse $\vc{x} \geq \vc{0}$:
\begin{equation}
  \{ \vc{z}\,|\, \mt{A}\vc{z}=\mt{A}\vc{x}\,\&\,\vc{z}\geq \vc{0}\}=\{\vc{x}\} \label{eq:uniqueness}
\end{equation}
At the center of their work is the notion of \emph{outwardly
  $s$-neighborly polytopes}. Assume w.l.o.g. that all columns $\vc{s}_j$ of $\mt{A}$
are non-zero and define their convex hull
\newcommand{\Pa}{P_{\mt{A}}}
\begin{equation*}
  \Pa:=\text{conv}(\vc{s}_1,\dots,\vc{s}_n).
\end{equation*}
This polytope is called \emph{$s$-neighborly} if every set of $s$
vertices spans a face of $\Pa$. If this is the case, the polytope $\Pa^0:=\text{conv}(\Pa\cup\{\vc{0}\})$
is called \emph{outwardly $s$-neighborly}. They then move on to prove that the solution to
\begin{equation*}
  \arg\min\lVert\vc{x}\rVert_{\ell_0}\quad\text{s.t.}\quad
  \mt{A}\vc{x}=\vc{y}
\end{equation*}
is unique if and only if $\Pa^0$ is outwardly $s$-neighborly \cite{Donoho2005}.

Another approach to the same question was introduced in Ref.~\cite{Bruckstein2008}. 
They consider full rank $m \times n$-matrices whose \emph{row space intersects the positive orthant}:
\begin{equation}
  \setM=\{\mt{A}\,:\,\exists\vc{t} \in \mathbb{R}^m\,\,\mt{A}^*\vc{t}>\vc{0}\}.
  \label{eq:def:Mset}
\end{equation}
Note that both structures are related in the sense that $\mt{A}\in\setM$, if and only if $\vc{0}\notin\Pa$ \cite{Wang2011}.
Also, a strictly positive row assures $\mt{A}\in\setM$. 
An extreme case thereof occurs if $\mt{A}$ contains the ``all-ones'' vector $\vc{1}_n$ in $\mathbb{R}^n$.
The corresponding measurement yields the $\ell_1$-norm 
$\lVert\vc{x}\rVert_{\ell_1}=\langle \vc{1}_n,\vc{x}\rangle$ and
therefore all admissible vectors in \eqref{eq:bpdn:nn} for $\eta=0$ have the same cost.
The
uniqueness property in such a setting has already been obtained by Fuchs
\cite{Fuchs2005} for Vandermonde measurement matrices and for
particular real Fourier measurements using convex duality.
In these special cases, $m$ distinct columns are linear
independent (``full spark'') and therefore Eq.~\eqref{eq:uniqueness} holds, provided that $\vc{x}$ is sufficiently sparse:
$\lVert\vc{x}^{(0)}\rVert_{\ell_0}\leq\frac{m-1}{2}$.

In Ref.~\cite{Bruckstein2008}, Bruckstein et al.\ investigated the recovery
of nonnegative vectors by \eqref{eq:bpdn:nn} and modifications of 
OMP using a coherence-based approach. They
obtained numerical evidence for unique recovery in the regime $s=\Order(\sqrt{n})$.
Later, Wang and coauthors \cite{Wang2011} have analyzed non-negativity
priors for vector and matrix recovery using an RIP-based analysis. Concretely,
they translated the well-known RIP-result of random
i.i.d.\ $\pm 1$-Bernoulli matrices (see for example \cite{Baraniuk2008})
to $0/1$-measurements in the following way. 
Perform measurements using  an 
$(m+1)\times n$ matrix $\mt{A}^1=\left(\vc{1}^T_n | \mt{A}^T \right)^T$
which consists of an all-ones row $\vc{1}_n$ appended by a random i.i.d.\ $0/1$-valued
$m\times n$ matrix $\mt{A}$. By construction, the first noiseless
measurement on a nonnegative vector $\vc{x}$ returns its
$\ell_1$-norm $\lVert\vc{x}\rVert_{\ell_1}=\langle \vc{1}_n,\vc{x}\rangle$.
Rescaling and subtracting this value from the $m$ remaining measurements then results in
$\pm1$-measurements. This insight allows for an indirect nullspace
characterization of $\mt{A}$ in terms of the restricted isometry property (RIP)
of i.i.d.~$\pm1$-Bernoulli random matrices $\tilde{\mt{A}}$.
Recall that a matrix obeys the RIP of order $s$, if it acts almost isometrically on $s$-sparse vectors:
There exists 
$\delta_s\in[0,1)$ such that
$|\lVert
\tilde{\mt{A}}\vc{x}\rVert_{\ell_2}^2-\lVert\vc{x}\rVert_{\ell_2}^2|\leq\delta_s\lVert\vc{x}\rVert_{\ell_2}^2$
for all $s$-sparse $\vc{x}$.
Cand\`es showed in \cite{candes:rip2008} that validity of a $2s$-RIP implies that a $(\ell_1,\ell_1)$-nullspace property is valid for each $\vc{v} \in \mathbb{R}^n$ that is contained in the nullspace $\mathcal{N}(\tilde{\mt{A}})$ of $\tilde{\mt{A}}$:
\begin{equation}
  \lVert
  \vc{v}_S\rVert_{\ell_1}\leq\frac{\sqrt{2}\delta_{2s}}{1-\delta_{2s}}
  \lVert \vc{v}_{\bar{S}}\rVert_{\ell_1}
\end{equation}
for all $\vc{v}\in\mathcal{N}(\tilde{\mt{A}})$ and support sets $S$ of size
$|S|\leq s$. Combining this
with $\mathcal{N}(\mt{A}^1)\subset\mathcal{N}(\tilde{\mt{A}})$ then allows for proving
unique recovery in regime $s=\Order(n)$ with overwhelming probability.

However, so far, all these results manifestly focus on noiseless measurements.
Thus, the robustness of these approaches towards noise corruption needs to be examined.
Foucart, for instance, considered the 
{\em $\ell_1$-squared nonnegative regularization} \cite{Foucart:NNLS:2014}:
\begin{equation}
  \min_{\vc{z}\geq \vc{0}}\lVert\vc{z}\rVert^2_{\ell_1}+\lambda^2\lVert \mt{A}\vc{z}-\vc{y}\rVert^2_{\ell_2}
  \label{eq:nnreg}
\end{equation}
which can be re-cast as nonnegative least-squares problem.
He then showed that for stochastic matrices\footnote{Recall that a matrix is stochastic, if all entries are non-negative and all columns sum up to one.} the solution of \eqref{eq:nnreg} converges to the solution of \eqref{eq:bpdn:nn}
for $\lambda \to \infty$. 

Here, we aim at establishing even stronger recovery guarantees that, among other things, 
require neither an a-priori noise bound $\eta$, nor a regularization parameter $\lambda$.
We have already mentioned that the quotient property would assure such bounds 
for Gaussian matrices in the optimal regime. But $m\times n$ Gaussian matrices fail
to be in $\setM$ with probability approaching one as long as
$\lim_{n\rightarrow} m/n<\frac{1}{2}$ \cite{Wang2011}. On the
algorithmic side, there exists variations of certain regression methods where the
regularization parameter can be chosen independent of the noise
power---see Ref.~\cite{giraud:unknownnoise} for more details on this topic.
For the LASSO selector, in particular, such modifications are known
as the  ``scaled LASSO'' and ``square root LASSO''
\cite{sun:slasso,Stadler2010a}.

Non-negativity as a further structural constraint has also been investigated in the statistics community.
But these works focus on the averaged
case with respect to (sub-)Gaussian additive noise, whereby we
consider instantaneous guarantees.  
Slawski and Hein \cite{Slawski2013}, as well as Meinshausen
\cite{Meinshausen2013} have recently investigated this averaged
setting.

Finally, we note that the measurement setup above using a separate
``all ones'' row can also casted as a linearly
constrained NNLS, i.e., minimizing $\lVert
\mt{A}\vc{x}-\vc{y}\rVert_{\ell_2}$ subject to $\vc{x}\geq \vc{0}$
\emph{and} $\langle \vc{1}_n,\vc{x}\rangle=\text{const.}$, see for
example \cite{Vila2014} for a Bayesian recovery approach.

\section{Nullspace Property with Nonnegative constraints}
\label{sec:nsp}

Throughout our work we endow $\mathbb{R}^n$ with the partial ordering
induced by the nonnegative orthant, i.e.\
$\vc{x} \leq \vc{z}$ if and only if $x_i \leq z_i$ for all $1 \leq i
\leq n$. Here, $x_i=\langle \vc{e}_i,\vc{x}\rangle$ are the components
of $\vc{x}$ with respect to the standard basis $\{\vc{e}_i\}_{i=1}^n$
of $\mathbb{R}^n$.
Similarly, we write $\vc{x} < \vc{z}$ if strict inequality holds in
each component. 
We also write $\vc{x} \geq \vc{0}$ to indicate that $\vc{x}$ is (entry-wise) nonnegative.
For $1\leq p\leq\infty$, we denote the vector $\ell_p$-norms by $\|\cdot\|_{\ell_p}$  and $\|\cdot\|$ is
the usual operator/matrix norm. 
The $\ell_1$-error of the best
$s$-term approximation of a vector $\vc{x}$ will be denoted by
$\sigma_s(\vc{x})_{\ell_1}$.

\subsection{The robust nullspace property} \label{sub:nsp_standard}

The implications of a NSP are by now well-established and can be found, for instance, in \cite[Sec.~4.3]{Foucart2013}. 
Suppose that a matrix $\mt{A}: \mathbb{R}^n \to \mathbb{R}^m$ obeys the $\ell_2$-robust nullspace property of order $s$ ($s$-NSP) from Definition~\ref{def:nsp}. 
Theorem~4.25 in \cite{Foucart2013} then states that
\begin{equation}
  \| \vc{x} - \vc{z} \|_{\ell_2} \leq 
  \frac{C}{\sqrt{s}} \left( \| \vc{z} \|_{\ell_1} - \| \vc{x} \|_{\ell_1}
    + 2 \sigma_s(\vc{x})_{\ell_1} \right) + D \tau \left\| \mt{A} (\vc{x} - \vc{z} ) \right\|_{\ell_2} 
  \label{eq:nsp_alternative}
\end{equation}
is true for any $\vc{x},\vc{z} \in \mathbb{R}^n$. Here, $C = \frac{(1+\rho)^2}{1-\rho}$ and $D= \frac{3+\rho}{1-\rho}$ depend only on the NSP parameter $\rho$.
Replacing $\vc{z}$ with the BPDN minimizer
$\vc{x}_\eta^\sharp$ from \eqref{eq:bpdn} for the sampling model
$\vc{y}=\mt{A}\vc{x}+\vc{e}$
then implies
\begin{align}
    \| \vc{x} - \vc{x}_\eta^\sharp \|_{\ell_2} 
    &\leq 
    \frac{2C}{\sqrt{s}}\sigma_s(\vc{x})_{\ell_1}  + D\tau \left\| \vc{y} - \vc{e}-\mt{A}\vc{x}_\eta^\sharp  \right\|_{\ell_2} 
    \leq 
    \frac{2C}{\sqrt{s}}\sigma_s(\vc{x})_{\ell_1}  + D\tau \left(\left\| \vc{y} - \mt{A}\vc{x}_\eta^\sharp  \right\|_{\ell_2} +\lVert\vc{e}\rVert_{\ell_2} \right) \nonumber \\
    &\leq 
    \frac{2C}{\sqrt{s}}\sigma_s(\vc{x})_{\ell_1}  + 2 D \tau \eta,
  \label{eq:nsp_alternative:recovery}
\end{align}
provided that $\lVert \vc{e} \rVert_{\ell_2}\leq\eta$ is true. This estimate follows from
exploiting
$\lVert \vc{x}^\sharp_\eta\rVert_{\ell_1}\leq \lVert
\vc{x}\rVert_{\ell_1}$ and 
and $\left\| \vc{y} - \mt{A}\vc{x}^\sharp_\eta \right\|_{\ell_2}\leq\eta$.
Evidently, it is only true for $\eta \geq \|\vc{e}\|_{\ell_2}$ which in turn requires \emph{some} knowledge about the noise corruption.

\subsection{Nonnegative Constraints} \label{sub:nsp_nn}

Here we will prove a variation of Formula~\eqref{eq:nsp_alternative}
which holds for nonnegative vectors and matrices in
$\setM\subset\mathbb{R}^{m\times n}$. For such matrices we define a condition number by
\begin{equation}
  \kappa(\mt{A})=\min\{\lVert\vc{W}\rVert\lVert \vc{W}^{-1}\rVert\,\,\,
  | \exists \vc{t}\in\mathbb{R}^m\,\,\text{with}\,\vc{W}=\diag(\mt{A}^T\vc{t})>0\}.
\end{equation}
Note that for diagonal matrices $\mt{W}$ with non-negative entries 
$\kappa(\vc{W})=\lVert\vc{W}\rVert\lVert \vc{W}^{-1}\rVert$.
\begin{theorem} \label{thm:main}
  Suppose that $\mt{A}\in\setM$ obeys the $s$-NSP with parameters $\rho$ and
  $\tau$, and let $\kappa=\kappa(\mt{A})$ be its condition number
  achieved for $\vc{t}\in\mathbb{R}^m$.
  If $\kappa\rho <1$, then
  \begin{align*}
    \| \vc{x} - \vc{z} \|_{\ell_2} 
    \leq & \frac{2 C'}{\sqrt{s}} \sigma_s
      (\vc{x})_{\ell_1}+ D' \left( \| \vc{t} \|_{\ell_2} + \tau \right)
      \left\| \mt{A} (\vc{x} - \vc{z} ) \right\|_{\ell_2}
  \end{align*}
  is true for all nonnegative vectors $\vc{x},\vc{z} \in \mathbb{R}^n$. The constants amount to
 $C' = \frac{\kappa(1+\kappa \rho)^2}{1-\kappa \rho}$ and $D' = \frac{3+\kappa \rho}{1-\kappa \rho} \max \left\{ \kappa, \|\mt{W}^{-1} \| \right\}$.
\end{theorem}

Comparing this to \eqref{eq:nsp_alternative} reveals that the $\ell_1$-term 
$(\| \vc{z} \|_{\ell_1} - \| \vc{x} \|_{\ell_1})$ is not present anymore.
Inserting
$\vc{y}=\mt{A}\vc{x}+\vc{e}$ and applying the triangle inequality results in
\begin{equation}
    \| \vc{x} - \vc{z}\|_{\ell_2} 
    \leq 
	\frac{2C}{\sqrt{s}} \sigma_s (\vc{x})_{\ell_1} + D \left( \| \vc{t} \|_{\ell_2} + \tau \right) \left( \| \mt{A} \vc{z} - \vc{y} \|_{\ell_2} + \| \vc{e} \|_{\ell_2} \right) \quad \forall \vc{x},\vc{z} \geq \vc{0}.
\label{eq:nsp_implication_nn}
\end{equation}
This observation already highlights that CS-oriented algorithms, which
typically minimize the $\ell_1$-norm, are not required anymore in
the non-negative case. Instead, in order to get good estimates it makes sense to minimize the r.h.s. of the bound over the ``free'' parameter $\vc{z} \geq \vc{0}$. 
Doing so, results in the non-negative least squares fit \eqref{eq:least_squares}.
The sought for vector $\vc{x}$ is itself a feasible point of this optimization problem and consequently $\| \mt{A} \vc{x}^\sharp - \vc{y} \|_{\ell_2} \leq \| \mt{A} \vc{x} - \vc{y} \|_{\ell_2} = \| \vc{e} \|_{\ell_2}$.
Inserting this into \eqref{eq:nsp_implication_nn} then implies
\begin{align*}
\| \vc{x} - \vc{x}^\sharp \|_{\ell_2}
\leq & \frac{2C}{\sqrt{s}} \sigma_s (\vc{x})_{\ell_1} + 2 D \left( \| \vc{t} \|_{\ell_2} + \tau \right) \| \vc{e} \|_{\ell_2}
\end{align*}
which is comparable to \eqref{eq:nsp_alternative:recovery}. However, rather than depending on an a-priori noise bound $\eta$, the reconstruction error scales proportionally to $\| \vc{e} \|_{\ell_2}$ itself.

We will require two auxiliary statements in order to prove Theorem~\ref{thm:main}:

\begin{lemma} \label{lem:nsp_propagation}
Suppose that $\mt{A}$ obeys the $s$-NSP with parameters $\rho$ and $\tau$, and set $\mt{W} = \diag(\vc{w})$, where $\vc{w} > \vc{0}$ is strictly positive.
  Then, $\mt{A}{\mt{W}^{-1}}$ also obeys the $s$-NSP
  with parameters $\tilde{\rho} = \kappa (\mt{W}) \rho$ and $\tilde{\tau} = \| \mt{W} \| \tau$. 
\end{lemma}

\begin{proof}
The fact that $\mt{W}$ is diagonal assures $\mt{W}^{-1} \vc{v}_S = \left( \mt{W}^{-1} \vc{v}\right)_S$ (same for $\bar{S}$). 
Also, $\mt{A}$ obeys the $s$-NSP by assumption. Consequently
  \begin{align*}
      \| \vc{v}_S \|_{\ell_2}
      &= \| \mt{W}\mt{W}^{-1} \vc{v}_S \|_{\ell_2} \leq \|\mt{W}\| \| (\mt{W}^{-1} \vc{v})_S \|_{\ell_2} 
       \leq \|\mt{W} \|\left(\frac{\rho}{\sqrt{s}} \|(\mt{W}^{-1}\vc{v})_{\bar{S}}\|_{\ell_2} + \tau \| \mt{A}\mt{W}^{-1} \vc{v} \|_{\ell_2}\right) \\
      & \leq \frac{\| \mt{W} \| \| \mt{W}^{-1}\| \rho}{\sqrt{s}} \| \vc{v}_{\bar{S}} \|_{\ell_1} + \| \mt{W}\| \tau \| \mt{A}\mt{W}^{-1} \vc{v} \|_{\ell_2}
= \frac{\tilde{\rho}}{\sqrt{s}} \| \vc{v}_{\bar{S}} \|_{\ell_1}+ \tilde{\tau} \| \mt{A}\mt{W}^{-1} \vc{v} \|_{\ell_2}
   \end{align*}
is true for every set $S$ with $|S| \leq s$.

\end{proof}

\begin{lemma} \label{lem:positivity_implication}
Fix $\mt{A} \in \mathbb{R}^{m \times n}$ and
suppose that $\vc{w}=\mt{A}^T\vc{t}$ is strictly positive for some $\vc{t}\in \mathbb{R}^m$. Also, set $\mt{W}=\mathrm{diag}(\vc{w})$. Then, the following relation holds for any pair of non-negative vectors $\vc{x},\vc{z} \geq \vc{0}$ in $\mathbb{R}^n$:
  \begin{equation*}
    \| \vc{W} \vc{z} \|_{\ell_1} - \| \vc{W} \vc{x} \|_{\ell_1} \leq
    \| \vc{t} \|_{\ell_2} \| \mt{A} \left( \vc{x} - \vc{z} \right)\|_{\ell_2}
  \end{equation*}
\end{lemma}

\begin{proof}
Note that, by construction, $\vc{W}$ is symmetric and preserves entry-wise non-negativity. These features together with positivity of $\vc{z}$ imply
\begin{equation*}\begin{split}
    \| \vc{W} \vc{z} \|_{\ell_1}
    =& \langle \vc{1}_n, \vc{W} \vc{z} \rangle = \langle \vc{W}\vc{1}_n,\vc{z}\rangle 
    = \langle \diag(\mt{A}^T \vc{t})\vc{1}_n,\vc{z} \rangle 
    = \langle \mt{A}^T \vc{t}, \vc{z} \rangle = \langle \vc{t}, \mt{A} \vc{z} \rangle.
\end{split}\end{equation*}
An analogous reformulation is true for $\| \mt{W} \vc{x}\|_{\ell_1}$
and combining these two reveals
\begin{equation*}
  \| \mt{W} \vc{z} \|_{\ell_1} - \|\mt{W} \vc{x} \|_{\ell_1}
  = \langle \vc{t}, \mt{A} \left( \vc{z} - \vc{x} \right) \rangle
  \leq \| \vc{t} \|_{\ell_2} \| \mt{A}(\vc{z}- \vc{x} ) \|_{\ell_2}
\end{equation*}
due to Cauchy-Schwarz.
\end{proof}

\begin{proof} [Proof of Theorem~\ref{thm:main}]
The assumption $\mt{A}\in\setM$ assures that there exists $\vc{t}\in \mathbb{R}^m$ such that $\vc{w}=\mt{A}^T\vc{t}>\vc{0}$ and we define $\mt{W}:= \diag (\vc{w})$. By assumption, $\mt{W}$ is
invertible and admits a condition number $\kappa= \| \mt{W} \|\|\mt{W}^{-1} \|$. 
Thus, we may write 
\begin{equation*}
  \| \vc{x} - \vc{z} \|_{\ell_2}
  = \| \mt{W}^{-1} \mt{W} \left( \vc{x} - \vc{z} \right) \|_{\ell_2} 
  \leq \| \mt{W}^{-1} \| \| \mt{W}(\vc{x} - \vc{z})\|_{\ell_2}
\end{equation*}
for any pair $\vc{x},\vc{z}>\vc{0}$.
Since $\mt{A}$ obeys the $s$-NSP, Lemma \ref{lem:nsp_propagation}
assures that $\mt{A}\mt{W}^{-1}$ also admits a $s$-NSP, albeit with parameters
$\tilde{\rho} = \kappa\rho$ and $\tilde{\tau} = \| \mt{W} \| \tau$. 
Thus, from  \eqref{eq:nsp_alternative} we conclude the following
for vectors $\mt{W}\vc{x}$ and $\mt{W}\vc{z}$:
\begin{align*}
    \|\mt{W}(\vc{x} - \vc{z})\|_{\ell_2}
    &\leq \frac{1}{\sqrt{s}}\frac{(1+\kappa \rho)^2}{1-\kappa \rho} \left( \|\mt{W} \vc{z} \|_{\ell_1} - \|\mt{W} \vc{x} \|_{\ell_1} + 2 \sigma_s(\mt{W} \vc{x})_{\ell_1}\right)
    +\frac{3+\kappa \rho}{1-\kappa \rho} \| \mt{W} \|\tau \|\mt{A}(\vc{x} - \vc{z})\|_{\ell_2} \\
    &\leq \frac{2(1+\kappa \rho)^2}{1-\kappa \rho}\| \mt{W}\| \frac{\sigma_s(\vc{x})_{\ell_1}}{\sqrt{s}}
    +\left(\frac{(1+\kappa \rho)^2}{1-\kappa \rho}\frac{\lVert\vc{t}\rVert_{\ell_2}}{\sqrt{s}}+\frac{3+\kappa \rho}{1-\kappa \rho} \| \mt{W}\|\tau \right) \| \mt{A} (\vc{x} - \vc{z})\|_{\ell_2} \\
 & \leq \frac{2(1+\kappa \rho)^2}{1-\kappa \rho}\| \mt{W}\| \frac{\sigma_s(\vc{x})_{\ell_1}}{\sqrt{s}}
    +\frac{3+\kappa \rho}{1-\kappa \rho}\left(\| \vc{t} \|_{\ell_2}+ \| \mt{W}\|\tau \right) \| \mt{A} (\vc{x} - \vc{z})\|_{\ell_2} 
\end{align*}
Here, we invoked Lemma
\ref{lem:positivity_implication}, 
as well as the relation 
$\sigma_s(\mt{W}\vc{x})_{\ell_1}\leq\lVert\mt{W}\rVert\sigma_s(\vc{x})_{\ell_1}$.
So, in summary we obtain
\begin{align*}
  \| \vc{x} - \vc{z} \|_{\ell_2}
\leq &\| \mt{W}^{-1} \| \| \mt{W}(\vc{x} - \vc{z})\|_{\ell_2} \\
\leq & \frac{2 \kappa (1+\kappa \rho)^2}{1-\kappa \rho} \frac{\sigma_s(\vc{x})_{\ell_1}}{\sqrt{s}}
    +\frac{3+\kappa \rho}{1-\kappa \rho}\left(\| \mt{W}^{-1} \| \|\vc{t}\|_{\ell_2}+ \kappa \tau \right) \| \mt{A} (\vc{x} - \vc{z})\|_{\ell_2} \\
\leq & \frac{2C'}{\sqrt{s}} \sigma(\vc{x})_{\ell_1}+ D' \left( \| \vc{t}\|_{\ell_2}+\tau \right) \| \mt{A}(\vc{x}-\vc{z}) \|_{\ell_2}
\end{align*}
with $C' = \frac{\kappa(1+\kappa \rho)^2}{1-\kappa \rho}$ and $D' = \frac{3+\kappa \rho}{1-\kappa \rho} \max \left\{ \kappa, \|\mt{W}^{-1} \| \right\}$.
\end{proof}

\section{Robust NSP for 0/1-Bernoulli matrices}
\label{sec:bin:model}

In this section, we prove our second main  result, Theorem~\ref{thm:main2}. 
Said statements summarizes two results, namely
(i) 0/1-Bernoulli matrices $\mt{A}$ with $m= C s \log (n/s)$ rows obey
the robust null space property of order $s$ w.h.p. and (ii) the row
space of $\mt{A}^T$ allows for constructing a strictly positive vector $\vc{w} = \mt{A}^T \vc{t} >\vc{0}$ (that is sufficiently well-conditioned). 
We will first state the main ideas and prove both statements in subsequent subsections.

\subsection{Sampling model and overview of main proof ideas}

Let us start by formally defining the concept of a 0/1-Bernoulli matrix. 

\begin{definition}
  We call $\mt{A}\in\mathbb{R}^{m\times n}$ a 0/1-Bernoulli matrix
  with parameter $p \in [0,1]$, if every matrix element
  $A_{i,j}$ of $\mt{A}$ is an independent realization of a
  Bernoulli random variable $b$ with parameter $p$, i.e.
  \begin{equation*}
    \mathrm{Pr} \left[b = 1 \right] = p \quad \textrm{and} \quad 
    \mathrm{Pr} \left[ b = 0 \right] = 1-p.
  \end{equation*}
\end{definition}

Recall that $\mathbb{E} \left[ b \right] = p$ and $\mathrm{Var}(b) = \mathbb{E} \left[ \left( b - \mathbb{E}[b] \right)^2 \right] = p(1-p)$. 
By construction, the $m$ rows $\vc{a}_1,\ldots,\vc{a}_m$ of such a
$0/1$-Bernoulli matrix are independent and obey
\begin{equation*}
  \mathbb{E} \left[ \vc{a}_k \right]
  = \sum_{j=1}^n \mathbb{E} \left[ A_{k,j} \right] \vc{e}_j = p \sum_{j=1}^n \vc{e}_j
  = p \vc{1}. 
\end{equation*}
This expected behavior of the individual rows will be crucial for addressing the second point in Theorem~\ref{thm:main2}:
Setting 
\begin{equation*}
\vc{w} := \frac{1}{pm} \sum_{k=1}^m \vc{a}_k = \mt{A}^T \left( \frac{1}{pm} \vc{1}_m \right)
\end{equation*}
results in a random vector $\vc{w} \in \mathbb{R}^n$ that obeys $\mathbb{E} \left[ \vc{w} \right] = \vc{1}>\vc{0}$. Applying a large deviation bound will in turn imply that a realization of $\vc{w}$ will w.h.p. not deviate too much from its expectation. This in turn ensures strict positivity.
We will prove this in Subsection~\ref{sub:strictly_positive}.

However, when turning our focus to establishing null space properties for $\mt{A}$,
working with 0/1-Bernoulli entries renders such a task more challenging.
The simple reason for such a complication is that the individual random entries of $\mt{A}$ 
are not centered, i.e. $\mathbb{E} \left[ A_{k,j} \right] = p \neq 0$. 
Combining this with independence of the individual entries yields
\begin{align*}
\mathbb{E} \left[ \vc{a}_k \vc{a}_k^T \right]
= p^2 \vc{1}_n \vc{1}^T_n + p(1-p) \mathbb{I}.
\end{align*}
This matrix admits a condition number of $\kappa \left( \mathbb{E} \left[ \vc{a}_k \vc{a}_k^T \right] \right) = 1 + \frac{pn}{1-p}$ which underlines the ensemble's anisotropy.
Traditional proof techniques, e.g. establishing an RIP, are either not applicable, or yield sub-optimal results \cite{RudelsonZhou13, KuengGross14}. 
This is not true for Mendelson's small ball method \cite{Mendelson15,KoltchinskiMendelson15} (see also \cite{Tropp2015})---a strong general purpose tool whose applicability only requires row-wise independence.
It was shown in Ref.~\cite{Dirksen16} that this technique allows for establishing the NSP for a variety of compressed sensing scenarios. Our derivation is inspired by the techniques presented in \cite{}\emph{loc.\ cit.}.
Moreover, a similar approach is applicable to the conceptually-related problem of low rank matrix reconstruction \cite{Kabanava2016}.

\subsection{Null Space Properties for $0/1$-Bernoulli matrices}

Recall that Definiton~\ref{def:nsp} states that a $m \times n$ matrix $\mt{A}$ obeys the robust null space property with parameters $\rho \in (0,1)$ and $\tau >0$, if
\begin{equation}
\| \vc{v}_S \|_{\ell_2} \leq \frac{\rho}{\sqrt{s}} \| \vc{v}_{\bar{S}} \|_{\ell_1} + \tau \| \mt{A} \vc{v} \|_{\ell_2} \label{eq:nsp}
\end{equation}
is true for all vectors $\vc{v} \in \mathbb{R}^n$ and support sets $S \in [n]$ with support size $|S| \leq s$. 
Demanding such  generality in the choice of the support set is in fact not necessary, see e.g. \cite[Remark 4.2]{Foucart2013}. For a fixed vector $\vc{v}$, the above condition holds for any index set $S$, if it holds for an index set $S_{\max}$ containing the $s$ largest (in modulus) entries of $\vc{v}$. Introducing the notation $\vc{v}_s := \vc{v}_{S_{\max}}$ and $\vc{v}_c := \vc{v}_{\bar{S}_{\max}}$, the robust null space property \eqref{eq:nsp} holds, provided that every vector $\vc{v} \in \mathbb{R}^n$ obeys
\begin{equation}
\| \vc{v}_s \|_{\ell_2} \leq \frac{\rho}{\sqrt{s}} \| \vc{v}_c \|_{\ell_1} + \tau \| \mt{A} \vc{v} \|_{\ell_2}. \label{eq:nsp2}
\end{equation}
Note that this requirement is invariant under re-scaling and we may w.l.o.g. assume $\| \vc{v} \|_{\ell_2} =1$. 
Moreover, for fixed parameters $s$ and $\rho$, any vector $\vc{v}$ obeying $\| \vc{v}_s \|_{\ell_2} \leq \frac{\rho}{\sqrt{s}} \| \vc{v}_c \|_{\ell_1}$ is guaranteed to fulfill \eqref{eq:nsp2} by default. Consequently, when aiming to establish null space properties,  it suffices to establish condition \eqref{eq:nsp2} for the set of unit-norm vectors that do not obey this criterion:
\begin{equation*}
T_{\rho,s}:= \left\{ \vc{v} \in \mathbb{R}^n: \; \| \vc{v} \|_{\ell_2} = 1, \; \| \vc{v}_s \|_{\ell_2} > \frac{\rho}{\sqrt{s}} \| \vc{v}_c \|_{\ell_1} \right\}.
\end{equation*}
As a result, a matrix $\mt{A}$ obeys the NSP \eqref{eq:nsp}, if
\begin{equation}
\inf \left\{ \| \mt{A} \vc{v} \|_{\ell_2}: \; \vc{v} \in T_{\rho,s} \right\} > \frac{1}{\tau}, \label{eq:nsp3}
\end{equation}
holds, where $\tau >0$ is the second parameter appearing in
\eqref{eq:nsp}.  For random $m \times n$ matrices $\mt{A}$ with
independent and identically distributed rows
$\vc{a}_1,\ldots,\vc{a}_m \in \mathbb{R}^n$ --- which is the case
here --- Mendelson's small ball method
\cite{Mendelson15,KoltchinskiMendelson15,Tropp2015} provides a general
purpose tool for establishing such lower bounds with high probability:

\begin{theorem}[Koltchinskii, Mendelson; Tropp's version \cite{Tropp2015}] \label{thm:mendelson}
  Fix $E \subset \mathbb{R}^n$ and let $\vc{a}_1,\ldots,\vc{a}_m$ be independent copies of a random vector $\vc{a} \in \mathbb{R}^n$. 
  Set $\vc{h} = \frac{1}{\sqrt{m}} \sum_{k=1}^m \epsilon_k \vc{a}_k$,
  where $\epsilon_1,\ldots,\epsilon_m$ is a Rademacher sequence. For $\xi >0$ define
  \begin{equation*}
    Q_\xi \left( E, \vc{a} \right) = \inf_{\vc{u} \in E} \mathrm{Pr} \left[  |\langle \vc{a}, \vc{u} \rangle | \geq \xi \right],
    \quad \textrm{as well as} \quad 
    W_m \left( E, \vc{a} \right) = \mathbb{E} \left[ \sup_{\vc{u} \in E}\langle \vc{h}, \vc{u} \rangle \right].
  \end{equation*}
  Then, for any $\xi >0$ and $t \geq 0$, the following is true with probability at least $1-\mathrm{e}^{-2t^2}$:
  \begin{equation}
    \inf_{\vc{v} \in E} \left( \sum_{k=1}^m \left| \langle \vc{a}_k, \vc{v} \rangle \right|^2 \right)^{1/2} \geq \xi \sqrt{m} Q_{2\xi} (E,\vc{a}) - \xi t - 2 W_m (E, \vc{a}) .
    \label{eq:mendelson}
  \end{equation}
\end{theorem}

In our concrete application, the random vector
$\vc{a}=\sum_{i=1}^n b_i \vc{e}_i \in \mathbb{R}^n$ has i.i.d.\
$0/1$-Bernoulli entries $b_i$ with parameter $p$ and $E= T_{\rho,r}$.  We
bound the marginal tail function
$Q_{2 \xi} \left( T_{\rho,s}, \vc{a} \right)$ from below using a
Paley-Zygmund inequality. Detailed in the appendix this calculation
yields
\begin{equation}
\mathrm{Pr} \left[ | \langle \vc{a},\vc{z} \rangle | \geq \theta \sqrt{p(1-p)}\right] \geq 
\frac{4}{13} p(1-p)(1-  \theta^2)^2 \quad \forall \vc{z} \in S^{n-1} \;\textrm{and} \; \theta \in [0,1]. \label{eq:Qxi}
\end{equation}
for any $\theta \in [0,1]$ and any $\vc{z} \in \mathbb{R}^n$ obeying $\| \vc{z} \|_{\ell_2}=1$. This, in particular includes any $\vc{z} \in T_{\rho,r}$ and consequently
\begin{equation*}
  Q_{2 \xi_0} \left( T_{\rho,s}, \vc{a}\right) \geq \frac{4p(1-p) (3/4)^2}{13} > \frac{p(1-p)}{6} \quad \textrm{for} \quad \xi_0 = \frac{1}{4}\sqrt{p(1-p)}.
\end{equation*}
In order to bound the mean empirical width $W_m \left( T_{\rho,r}, \vc{a} \right)$, we follow the approach outlined in Ref.~\cite{Dirksen16}. 
 Note that $T_{\rho,s}$ contains the set of all $s$-sparse vectors with unit length:
\begin{equation}
\Sigma_s^2 = \left\{ \vc{v} \in \mathbb{R}^n:\; \|\vc{v}\|_{\ell_0}\leq s,\; \| \vc{v} \|_{\ell_2}=1 \right\}.
\end{equation}
Several existing results, such as \cite[Lemma~3]{Kabanava2015} (see also \cite[Lemma~4.5]{Rudelson2008} and \cite[Lemma~3.2]{Dirksen16} for a generalization to normalization in any $\ell_q$-norm), state that a converse relation is also true:
\begin{equation}
T_{\rho,s} \subset \sqrt{1+(1+1/\rho)^2} \mathrm{conv} \left( \Sigma_s^2 \right) \subseteq \frac{3}{\rho} \mathrm{conv} \left( \Sigma_s^2 \right).
\end{equation}
Here, $\mathrm{conv} \left( \Sigma_s^2 \right)$ denotes the convex hull of $\Sigma_s^2$. This in turn implies
\begin{align*}
W_m \left( T_{\rho, s}, \vc{a} \right)=& \mathbb{E} \left[  \sup_{\vc{u} \in T_{\rho,r}} \langle \vc{h},\vc{u} \rangle \right] \leq \frac{3}{\rho} \mathbb{E} \left[ \sup_{\vc{u} \in \mathrm{conv}(\Sigma_s^2)} \langle \vc{h},\vc{u} \rangle \right] 
= \frac{3}{\rho} W_m \left( \Sigma_s^2, \vc{a} \right) ,
\end{align*}
where the last equation is due to the fact that the supremum of the linear function $\langle \vc{u},\vc{h} \rangle$ over the convex set $\mathrm{conv} \left( \Sigma_s^2\right)$ is attained at its extremal set $\Sigma_s^2$.
The quantity $W_m \left( \Sigma_s^2, \vc{a} \right)$ corresponds to the supremum of the stochastic process $X_{\vc{u}} = \langle \vc{u},\vc{h} \rangle$ indexed by $\vc{u} \in \Sigma_s^2$.
This stochastic process is centered ($\mathbb{E} \left[ X_{\vc{u}} \right] = 0$) and inherits subgaussian marginals from the fact that the individual entries of $\vc{a}$ are subgaussian random variables. 
Dudley's inequality, see e.g.\ \cite[Sec.~8.6]{Foucart2013}, allows for bounding the supremum of such centered, subgaussian stochastic processes. A computation detailed in the appendix yields
\begin{equation}
W_m \left( \Sigma_s^2, \vc{a} \right)
\leq 20 \theta (p) \sqrt{ s\left( \log \left( \frac{ \mathrm{e}{n}}{s} \right) + \frac{p^2}{\theta^2(p)} \right)},
\quad \textrm{where} \quad  
\theta (p) = \sqrt{\frac{2p-1}{2 \log \left( \frac{p}{1-p} \right)}}
\label{eq:Wm_bound}
\end{equation}
is the subgaussian parameter associated with the centered Bernoulli random variable $\tilde{b}$ with parameter $p$ \cite{Buldygin13}: $\mathrm{Pr} \left[ \tilde{b} = 1-p \right] = p$ and $\mathrm{Pr} \left[ \tilde{b}=-p \right] = 1-p$.

Fixing $t_0 = \frac{p(1-p)}{12} \sqrt{m}$ and inserting these bounds into Formula~\eqref{eq:mendelson} reveals
\begin{align*}
\inf_{\vc{v} \in T_{\rho,s}} \left\| \mt{A} \vc{v} \right\|_{\ell_2}
\geq & \xi_0 \sqrt{m} Q_{2\xi_0} \left( T_{\rho,s}, \vc{a}\right) - \xi_0 t_0 - 2 W_m \left( T_{\rho,s}, \vc{a} \right) \\
\geq & \frac{1}{48} \sqrt{p(1-p)}^3 \sqrt{m} - 120 \frac{\theta (p)}{\rho} \sqrt{ s\left( \log \left( \frac{ \mathrm{e}{n}}{s} \right) + \frac{p^2}{\theta^2(p)} \right)} 
\end{align*}
with probability at least $1- \mathrm{e}^{-\frac{1}{72}p^2(1-p)^2 m}$. In order to assure strict positivity of this bound, we set
\begin{equation*}
m \geq C_1 \frac{2 \theta^2 (p)}{p^3 (1-p)^3 \rho^2} s \left( \log \left( \frac{ \mathrm{e} n}{s} \right) + \frac{p^2}{\theta^2(p)} \right),
\end{equation*}
where $C_1 >0$ is a sufficiently large constant. Then the inequality above assures that there is another constant $C_2 >0$ (whose size only depends on $C_1$) such that
\begin{equation*}
\inf_{\vc{v} \in T_{\rho,s}} \geq \frac{1}{C_2} \sqrt{p(1-p)}^3 \sqrt{m}.
\end{equation*}
Comparing this bound to Eq.~\eqref{eq:nsp3} allows us to set $\tau = \frac{C_2}{\sqrt{p(1-p)}^3 \sqrt{m}}$. This is the main result of this section:

\begin{theorem} \label{thm:nsp} 
Let $\mt{A} \in \mathbb{R}^{m \times n}$ be a $0/1$-Bernoulli matrix with parameter $p \in [0,1]$.
 Fix $s \leq n$ and $\rho \in [0,1]$ and set
\begin{equation}
m = \frac{C_1}{\rho^2} \alpha (p) s \left( \log \left( \frac{ \mathrm{e} n}{s} \right) + \beta (p) \right)
\label{eq:sampling_rate}
\end{equation}
with $\alpha (p) = \frac{2p-1}{p^3(1-p)^3 \log \left( \frac{p}{1-p} \right)}$ and $\beta (p) = \frac{2 p^2 \log \left( \frac{p}{1-p}\right)}{2p-1}$.
Then, with probability of failure bounded by $\mathrm{e}^{-\frac{p^2(1-p)^2}{72} m}$, $\mt{A}$ obeys the robust NSP of order $s$ with parameters $\rho$ and $\tau = \frac{C_2}{\sqrt{p(1-p)}^3 \sqrt{m}}$.
Here, $C_1,C_2 >0$ denote absolute constants that only depend on each other.
\end{theorem}

This is a more detailed version of the first claim presented in Theorem~\ref{thm:main2}. 
We see that the sampling rate, the size of the NSP-parameter $\tau$ and the probability bound all depend on the Bernoulli parameter $p \in [0,1]$. 

\begin{remark} \label{rem:parameter}
While the sampling rate \eqref{eq:sampling_rate} is optimal in terms of sparsity $s$ and problem dimension $n$, this is not the case for its dependence on the parameter $p$. 
In fact, the first version of this work (e.g. see \cite{kueng:itw16})
achieved a strictly better constant $\tilde{\alpha} (p) = \frac{1}{p^2(1-p)^2}$ at the cost of a sub-optimal sampling rate of order $s\log(n)$. 
However, we do not know if this result accurately describes the correct behavior for the practically relevant case of sparse (co-sparse) measurements $p \to 0$ ($p \to 1$). 
We intend to address this question in future work.
\end{remark}

Finally, we point out that when opting for a standard Bernoulli process, i.e.\ $p = \frac{1}{2}$, 
the assertions of Theorem~\ref{thm:nsp} considerably simplify:

\begin{corollary}
Fix $s \leq n$, $\rho \in [0,1]$ and let $\mt{A}$ be a standard $(m \times n)$ $0/1$-Bernoulli matrix (i.e. $p=\frac{1}{2}$) with
\begin{equation*}
m \geq \frac{C_1}{128\rho^2} s\log (n).
\end{equation*}
Then with probability at least $1- \mathrm{e}^{-\frac{m}{1152}}$ this matrix obeys the NSP of order $s$ with parameters $\rho$ and $\tau = \frac{C_2}{8\sqrt{m}}$. 
Here, $C_1$ and $C_2$ are the constants from Theorem~\ref{thm:nsp}.
\end{corollary}

\subsection{0/1-Bernoulli matrices lie in $\mathcal{M}_+$} \label{sub:strictly_positive}

We now move on to showing that $0/1$-Bernoulli matrices are very likely to admit the second requirement of Theorem~\ref{thm:main}.
Namely, that there exists a vector $\vc{w} = \mt{A}^T \vc{t}$ that is
strictly positive (this is equivalent to demanding $\mt{A} \in
\mathcal{M}_+$). 
Concretely, we show that setting $\vc{t} = \frac{1}{pm}\vc{1}_m$ w.h.p. results in a strictly positive vector $\vc{w} \in \mathbb{R}^n$ whose conditioning obeys
\begin{equation}
\kappa (\vc{w} ) = \frac{ \max_k | \langle \vc{e}_k, \vc{w} \rangle |}{\min_k | \langle \vc{e}_k, \vc{w} \rangle | } \leq 3. \label{eq:kappa_bound}
\end{equation}
To do so, we note that $\vc{w} = \frac{1}{pm} \sum_{k=1}^m \vc{a}_k$
has expectation $\mathbb{E} \left[ \vc{w} \right] = \vc{1}_n$, which is---up to re-scaling---the unique non-negative vector admitting
$\kappa (\vc{1}_n) = 1$. 
After having realized this, it suffices to use a concentration inequality to prove that w.h.p. $\vc{w}$ does not deviate too much from its expectation. We do this by invoking a large deviation bound.

\begin{theorem} \label{thm:stictly_positive_vector}
  Suppose that $\mt{A}: \mathbb{R}^n \to\mathbb{R}^m$ is a $0/1$-Bernoulli matrix with parameter $p \in [0,1]$
  and set 
  \begin{equation}
    \vc{w} = \mt{A}^T \vc{t} \in \mathbb{R}^n \quad \textrm{with} \quad
    \vc{t} = \frac{1}{pm} \vc{1}_m \in \mathbb{R}^m.
  \end{equation}
  Then with probability at least $1- n \mathrm{e}^{-\frac{3}{8}p(1-p)m}$ 
  \begin{equation}
    \max_i | \langle \vc{e}_i, \vc{w} \rangle | \leq \frac{3}{2}
    \quad\text{and}\quad    
    \min_i | \langle \vc{e}_i, \vc{w} \rangle | \geq \frac{1}{2}.
  \end{equation}
  This in turn implies \eqref{eq:kappa_bound}.
\end{theorem}

\begin{proof}
  Instead of showing the claim directly, we prove the stronger statement:
  \begin{equation}
    | \langle \vc{e}_i, \vc{w} \rangle - 1 |  \leq \frac{1}{2} \quad 1 \leq i \leq n,
    \label{eq:kappa_aux1}
  \end{equation}
  is true with probability of failure bounded by $n
  \mathrm{e}^{-\frac{3}{8}p(1-p)m}$. If such a bound is true 
  for all $i$, it is also valid for maximal and minimal vector components and we obtain
  \begin{equation*}
    \max_i | \langle \vc{e}_i, \vc{w} \rangle | \leq  \max_{k} | \langle \vc{e}_i,\vc{w} \rangle -1 | +1 \leq  \frac{3}{2} \quad \textrm{and} \quad
    \min_k | \langle \vc{e}_i,\vc{w} \rangle |
    \geq  1 - \max_i | \langle \vc{e}_i, \vc{w} \rangle - 1 | \geq \frac{1}{2},
  \end{equation*}
  as claimed. 
  In order to prove \eqref{eq:kappa_aux1}, we fix $1 \leq i \leq n$ and focus on
  \begin{equation*}
    | \langle \vc{e}_i, \vc{w} \rangle -1 |
    = \left| \frac{1}{pm} \sum_{k=1}^m \langle \vc{e}_i, \vc{a}_k \rangle - 1 \right| \\
    = \frac{1}{pm} \left| \sum_{k=1}^m \left( b_{k,i} - \mathbb{E} \left[b_{k,i} \right] \right) \right|.
  \end{equation*}
  Here, we have used
  $\langle \vc{e}_i, \vc{a}_k \rangle = \langle \vc{e}_k, \mt{A}
  \vc{e}_i \rangle = b_{k,i}$,
  which is an indepenent instance of a $0/1$-Bernoulli random variable with
  parameter $p$.  Thus we are faced with bounding the deviation of a
  sum of $m$ centered, independent random variables
  $c_k := b_{k,i} - \mathbb{E} \left[ b_{k,i} \right]$ from its mean.
  Each such variable obeys
  \begin{equation*}
    | c_k| \leq   \max \left\{p,1-p \right\} \leq 1 \quad \textrm{and} \quad
    \mathbb{E} \left[c_k^2 \right] = \mathrm{Var}(b_{k,i}) = p(1-p).
  \end{equation*}
  Applying a Bernstein inequality \cite[Theorem 7.30]{Foucart2013} reveals
  \begin{equation*}
    \mathrm{Pr} \left[ | \langle \vc{e}_i, \vc{w} \rangle -1 | \geq \frac{1}{2} \right]
    \leq  \mathrm{Pr} \left[ | \langle \vc{e}_i,\vc{w} \rangle -1 | \geq \frac{1-p}{2} \right] 
    = \mathrm{Pr} \left[ \left| \sum_{k=1}^m c_k \right| \geq \frac{mp(1-p)}{2} \right] 
    \leq  \exp \left( - \frac{3}{8} p(1-p) m \right).
  \end{equation*}
  Combining this statement with a union bound assures that $| \langle \vc{e}_i,\vc{w} \rangle - 1 | < \frac{1}{2}$ is simultaneously true for all $1 \leq i \leq n$ with probability at least $1- n \mathrm{e}^{-\frac{3}{8}p(1-p)m}$.
\end{proof}

\subsection{Proof of Theorem~\ref{thm:main2}}

Finally, these two results can be combined to yield Theorem~\ref{thm:main2}. It readily follows from taking a union bound over the individual probabilities of failure.
Theorem~\ref{thm:nsp} requires a sampling rate of 
\begin{equation}
m = \frac{C_1}{\rho^2} \alpha (p) s \left( \log \left( \frac{ \mathrm{e} n}{s} \right) + \beta (p) \right)
\label{eq:sampling_rate2}
\end{equation} to assure that a corresponding $0/1$-Bernoulli matrix obeys a strong version of the NSP with probability at least $1- \mathrm{e}^{-\frac{p^2(1-p)^2}{72}m}$.
On the other hand, Theorem~\ref{thm:stictly_positive_vector} asserts that choosing $\vc{w}=\mt{A}^T \frac{1}{pm} \vc{1}_m$ for $0/1$-Bernoulli matrices $\mt{A}$ results in a well-conditioned and strictly positive vector $\vc{w}$ with probability at least $1- n \mathrm{e}^{-\frac{3}{8}p(1-p)m}$. 
The probability that either of these assertions fails to hold can be controlled by the union bound over both probabilities of failure:
\begin{align*}
\mathrm{e}^{-\frac{p^2(1-p)^2}{72}m} + n \mathrm{e}^{-\frac{3p(1-p)}{8}m} 
\leq  (n+1) \mathrm{e}^{-\frac{p^2(1-p)^2}{72}m}.
\end{align*}

Finally, we focus on $0/1$-Bernoulli matrices $\mt{A}$ for which both statements are true and whose sampling rate exceeds \eqref{eq:sampling_rate2}. 
Theorem~\ref{thm:nsp} then implies that $\mt{A}$ obeys the $s$-NSP with a pre-selected parameter $\rho \in [0,1]$ and $\tau = \frac{C_2}{\sqrt{p(1-p)}^3 \sqrt{m}}$. 
Moreover, the choice $\vc{t} = \frac{1}{pm} \vc{1}_m$ in Theorem~\ref{thm:stictly_positive_vector} results in $\| \vc{t} \|_{\ell_2} = \frac{1}{p \sqrt{m}}$. 
As a result, Theorem~\ref{thm:main} implies the following for any $\vc{x},\vc{z} \geq \vc{0}$:
\begin{align*}
\| \vc{x} - \vc{z} \|_{\ell_2}
\leq & \frac{2C'}{\sqrt{s}} \sigma_s (\vc{x} )_{\ell_1} + D' \left( \| \vc{t} \|_{\ell_2} + \tau \right) \| \mt{A} (\vc{x}-\vc{z} ) \|_{\ell_2} \\
=& \frac{2C'}{\sqrt{s}} \sigma_s (\vc{x})_{\ell_1} + D' \left( \frac{1}{p\sqrt{m}} + \frac{C_2}{\sqrt{p(1-p)}^3 \sqrt{m}} \right) \| \mt{A}(\vc{x}-\vc{z} ) \|_{\ell_2} \\
\leq & \frac{2C'}{\sqrt{s}} \sigma_s (\vc{x})_{\ell_1} + \frac{D'(1+C_2)}{\sqrt{p(1-p)}^3} \frac{\| \mt{A} (\vc{x}-\vc{z} ) \|_{\ell_2}}{\sqrt{m}}.
\end{align*}
Setting $\vc{z}=\vc{x}^\sharp \geq\vc{0}$ to be the solution of NNLS \eqref{eq:least_squares}
in turn assures
\begin{equation*}
\| \mt{A} \left( \vc{x}-\vc{x}^\sharp \right) \|_{\ell_2}
\leq \| \vc{e} \|_{\ell_2} + \| \vc{y}-\vc{x}^\sharp \|_{\ell_2}\leq 2 \| \vc{e}\|_{\ell_2}
\end{equation*}
and the claim follows with $E' = 2 D'(1+C_2)$.

\section{Numerical Experiments}
\label{sec:numeval}

This section is devoted to numerical tests regarding the \emph{nonnegative least squares} (NNLS)
estimation in
\eqref{eq:least_squares}.  To benchmark it, we compare this to the results obtained with {\em basis pursuit denoising} (BPDN)
in \eqref{eq:bpdn}. The NNLS has been computed using the {\tt
  lsqnonneg} function in {\sc Matlab} which implements the ``active-set''
Lawson-Hanson algorithm \cite{Lawson74}. For the BPDN the 
{\sc SPGL1} toolbox has been used \cite{BergFriedlander:2008}.

In a first test we have evaluated numerically the phase transition 
of NNLS in the $0/1$-Bernoulli setting for
the noiseless case. 
The dimension and sparsity parameters are
generated uniformly (in this order) in the ranges $n\in[10\dots500]$,
$m\in[10\dots n]$ and $s\in[1\dots m]$. Thus, the sparsity/density
variable is $r=s/m$ 
and the sub-sampling ratio is $\delta=m/n$.
The $m\times n$ measurement
matrix $\mt{A}$ is generated using the i.i.d.\ $0/1$-Bernoulli model with
$p=1/2$. The nonnegative $s$-sparse signal
$\vc{0}\le\vc{x}\in\mathbb{R}^n$ to recover is created as follows: the
random support $\supp(\vc{x})$ is obtained from taking the first $s$ elements
of a random (uniformly-distributed) permutation of the indices $(1\dots n)$.
On this support each component is the absolute value of an i.i.d.\ standard
(zero mean, unit variance) Gaussian, i.e., $x_i=|g_i|$ with $g_i\sim
N(0,1)$ for all $i\in\supp(\vc{x})$. We consider one individual recovery to be successful
if $\lVert \vc{x}-\vc{\hat{x}}\rVert_{\ell_2}\leq 10^{-3}\lVert \vc{x}\rVert_{\ell_2}$.
The resulting phase transition diagram, shown in Figure~\ref{fig:nnls:noiseless}
above, demonstrates that NNLS indeed reliable recovers nonnegative
sparse vectors without any $\ell_1$-regularization.

\begin{figure*}[ht]
\centering
\subfloat[][]{
  \includegraphics[width=.5\linewidth]{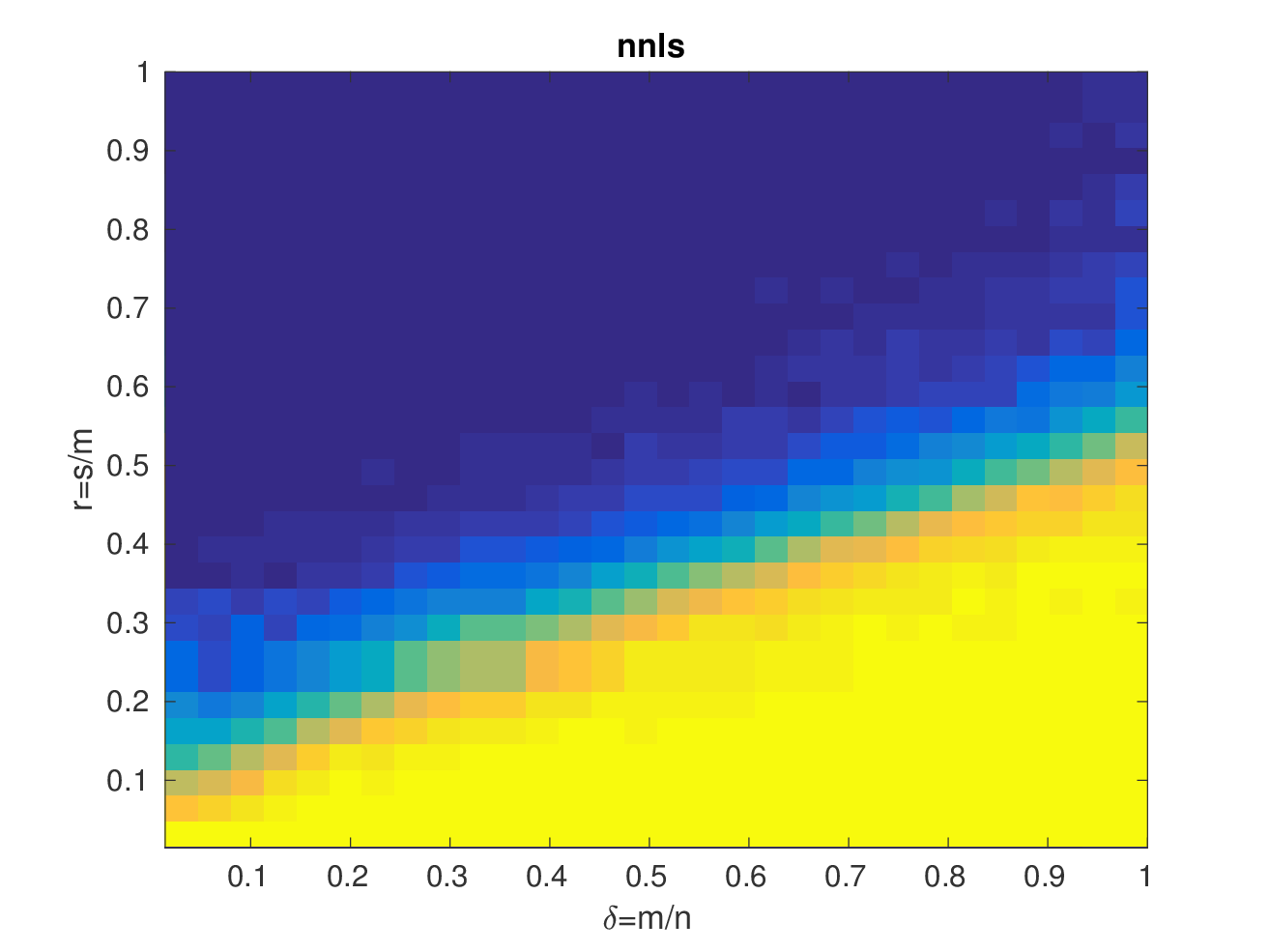} 
}
\subfloat[][]{
  \includegraphics[width=.5\linewidth]{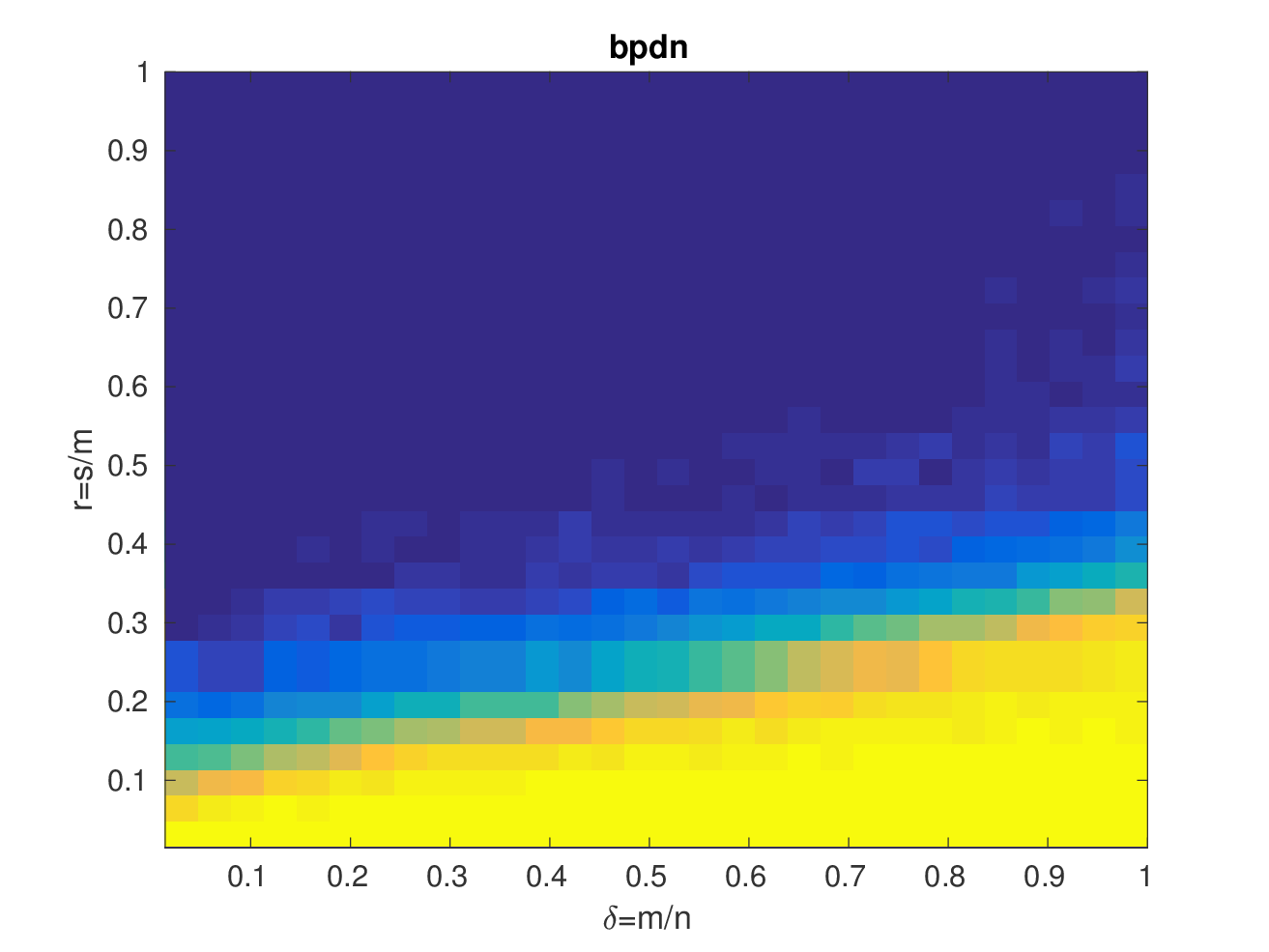} 
}
\caption{Comparison of NNLS in \eqref{eq:least_squares} 
  with BPDN in \eqref{eq:bpdn}
  for 
  i.i.d.\ $0/1$-Bernoulli matrices in the noisy setting. 
} 
\label{fig:nnls:bpdn:noisy}
\end{figure*}

In the second experiment we consider the noisy case. 
Apart from its simplicity, the important feature of NNLS is that no a-priori norm assumptions on the noise are
necessary. This is not the case for BPDN.
Theorem~\ref{thm:main2} implies that
the NNLS estimate $\vc{x}^\sharp$ obeys
\begin{equation}
  \| \vc{x} - \vc{x}^\sharp \|_{\ell_2} \leq \frac{D'}{8\sqrt{m}} \| \vc{e} \|_{\ell_2}
  \label{eq:numeval:errcut:noisy}
\end{equation}
A similar bound is valid for the BPDN (see
\eqref{eq:nsp_alternative:recovery}) estimate $\vc{x}_\eta$ provided that
$\|\vc{e} \|_{\ell_2}\leq\eta$. Note, that achieving this requires
knowledge of $\|\vc{e} \|_{\ell_2}$.
Interestingly, even under this prerequisite (BPDN indeed uses here the
instantaneous norm $\eta:=\|\vc{e} \|_{\ell_2}$ of the noise)
the performance of NNLS
is considerably better then BPDN in our setting.
This is visualized
in Figure~\ref{fig:nnls:bpdn:noisy} where each component $e_j$ of
$\vc{e}$ is i.i.d.\ Gaussian distributed with zero mean and variance
$\sigma^2_e=1/100$.  There recovery has been identified as
``successful'' if \eqref{eq:numeval:errcut:noisy} is fulfilled for
$\frac{D'}{8}=\sqrt{10}$.

Finally, we show in Figure \ref{fig:nnls:gaussiancomparison} that NNLS
is not as well-suited for uniform recovery with Gaussian measurements.
We have considered the noiseless scenario and generated random
$m\times n$ i.i.d.\ Gaussian and $0/1$-Bernoulli matrices where
$n=100$ and $m=20\dots80$. For each generated matrix we have tested
$10000$ random $s$-sparse vectors with $s=5$ as explained above.  We
counted an event as successful, if all $10000$ test vectors were
recovered within the bound 
$\lVert \vc{x}-\vc{x}^\sharp\rVert_{\ell_2}\leq 10^{-3}\lVert
\vc{x}\rVert_{\ell_2}$. We repeated this procedure $200$ times and
accumulated the results for every $m$.
\begin{figure*}[ht]
\centering
\includegraphics[width=.7\linewidth]{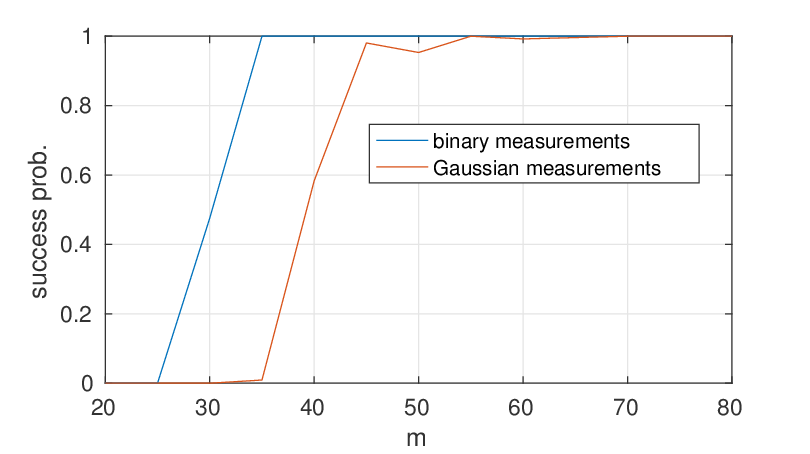} 

\caption{Comparison of NNLS in \eqref{eq:least_squares}  
  for i.i.d.\ Gaussian and $0/1$-Bernoulli matrices  
} 
\label{fig:nnls:gaussiancomparison}
\end{figure*}

\section{Conclusions}
In this work we have shown that non-negativity is an 
important additional property when recovering sparse vectors. 
This additional structural constraint is relevant in
many applications. Here, we provided activity detection in wireless
networks using individual sequences as concrete example.
There, designing measurement matrices such that convex hull of its columns
(the sequences) is sufficiently well-separated from the origin allow for 
remarkably simple and robust recovery algorithms. Crucially, these are robust to noise and blind
in a sense that no regularization and a-priori information on the noise
is required. We have demonstrated this feature by strengthening the
implications of the robust nullspace property for the non-negative setting.
Furthermore, we have proved that i.i.d.\ binary measurements fulfill w.h.p.
this property and are simultaneously well-conditioned. Therefore, they can be used
for recovering nonnegative and sparse vectors in the optimal regime.

\section*{Acknowledgement}
The authors want to thank S.\ Dirksen, H.\ Rauhut and D.\ Gross for inspiring discussions and helpful comments.
We also thank the anonymous reviewers for their constructive
comments, which helped us to improve the manuscript.
This work has been been supported by ``Mathematics of Signal
Processing'' trimester program at the Hausdorff Research Institute for
Mathematics (HIM). RK acknowledges  support from the Excellence Initiative
of the German Federal and State Governments (Grants
ZUK 43 \& 81), the ARO under contract W911NF-14-1-0098 (Quantum Characterization, Verification, and Validation), and the DFG.
PJ is supported by DFG grant JU 2795/2\&3.

\section*{Appendix}

Here we provide derivations of the two bounds \eqref{eq:Wm_bound} and \eqref{eq:Qxi} on which we built our argument that $0/1$-Bernoulli matrices obey the robust NSP. 
Since both are rather technical and not essential for understanding the main ideas, we decided to present them in this appendix.

\subsection{An upper bound on $W_m \left( \Sigma_s^2, \vc{a} \right)$ for $0/1$-Bernoulli matrices} \label{sub:H}

We will follow Ref.~\cite{Dirksen16} and use Dudley's inequality to bound the mean empirical width $W_m \left( \Sigma_s^2,\vc{a} \right)$ in Mendelson's small ball method; see also Ref.~\cite{Tropp2015} for a similar approach. 
Recall that $\vc{a}=\sum_{i=1}^n b_i \vc{e}_i \in \mathbb{R}^n$ is a random vector whose entries are i.i.d.\ Bernoulli random variables with parameter $p$.
We decompose $\vc{a}$ into $\tilde{\vc{a}}+ p \vc{1}$, where each entry $\tilde{b}_i$ of $\tilde{\vc{a}}$ is an i.i.d.\ copy of the centered Bernoulli random variable 
\begin{equation}
\tilde{b}
= \begin{cases}
1-p & \textrm{with prob. } p, \\
-p & \textrm{with prob. } 1-p.
\end{cases}
\label{eq:centered_bernoulli}
\end{equation}
Likewise, we introduce $\tilde{\vc{h}} = \frac{1}{\sqrt{m}} \sum_{k=1}^m \epsilon_k \tilde{\vc{a}}_k$ and note that
\begin{align*}
W_m \left( \Sigma_s^2, \vc{a} \right) =& \mathbb{E} \left[ \sup_{\vc{u} \in \Sigma_s^2} \langle \vc{u}, \vc{h} \rangle \right]
= \mathbb{E} \left[ \sup_{\vc{u} \in \Sigma_s^2} \langle \vc{u}, \tilde{\vc{h}}+\frac{p}{\sqrt{m}} \sum_{k=1}^m \epsilon_k \vc{1}_n \rangle \right] \\
\leq & \mathbb{E} \left[ \sup_{\vc{u} \in \Sigma_s^2} \langle \vc{u}, \tilde{\vc{h}} \rangle \right] + p \mathbb{E} \left[ \left| \frac{1}{\sqrt{m}} \sum_{k=1}^m \epsilon_k \right| \right] \sup_{\vc{u} \in \Sigma_s^2} \langle \vc{u}, \vc{1}_n \rangle \\
=& W_m \left( \Sigma_s^2, \tilde{\vc{a}} \right) + p \sqrt{\frac{s}{m}} \mathbb{E} \left[ \left|  \sum_{k=1}^m \epsilon_k \right| \right],
\end{align*}
because $\langle \vc{u},\vc{1}_n \rangle \leq \| \vc{u} \|_{\ell_1} \leq \sqrt{s} \| \vc{u} \|_{\ell_2} = \sqrt{s}$ for any $\vc{u} \in \Sigma_s^2$ (and this chain of inequalities is tight). The second term in this expression can be bounded via a Khintchine-type inequality. Corollary~8.7 in \cite{Foucart2013} implies (with $q=1$ and $\vc{c}=\vc{1}_m \in \mathbb{R}^m$)
\begin{equation*}
p \sqrt{\frac{s}{m}} \mathbb{E} \left[ \left|  \sum_{k=1}^m \epsilon_k \right| \right]
\leq \frac{\sqrt{s} p 2^{3/4}\mathrm{e}^{-1/2} \| \vc{1}_m \|_{\ell_2}}{\sqrt{m}} \leq \sqrt{2s} p.
\end{equation*}

The remaining term $W_m \left( \Sigma_s^2, \tilde{\vc{a}} \right)$ corresponds to the supremum of the stochastic process $\tilde{X}_\vc{u} = \langle \vc{u}, \tilde{\vc{h}} \rangle$ which is indexed by $\vc{u} \in \Sigma_s^2$.
This process is centered ($\mathbb{E} \left[ X_{\vc{u}} \right] = 0$ $\forall \vc{u} \in \Sigma_s^2$) and also subgaussian. 
A centered random variable $X$ is \emph{subgaussian} with parameter $\theta$, if its moment generating function obeys
\begin{equation}
\mathbb{E} \left[ \mathrm{e}^{\lambda X} \right] \leq \mathrm{e}^{\frac{1}{2}\theta^2 \lambda^2} \quad \forall \lambda \in \mathbb{R} \label{eq:subgaussian_definition}.
\end{equation}
We refer to \cite[Sec.~7.4]{Foucart2013} and \cite[Sec.~5.2.3]{Vershynin11} for a thorough introduction to subgaussian random variables. Here, we content ourselves with stating that \eqref{eq:subgaussian_definition} implies
\begin{equation*}
\mathrm{Pr} \left[ |X| \geq t \right] \leq 2 \mathrm{e}^{-\frac{t^2}{2\theta^2}} \quad \forall t \geq 0,
\end{equation*}
see e.g.\ \cite[Proposition 7.24]{Foucart2013}. Thus, every random variable obeying \eqref{eq:subgaussian_definition} has a subgaussian tail behavior governed by $\theta^2$. 
The centered random variable $\tilde{b}$ is a particular instance of a subgaussian random variable. The exact value of its subgaussian parameter has been determined in Ref.~\cite{Buldygin13}:
\begin{equation*}
\theta (p) = \sqrt{\frac{2p -1}{2 \log \left( \frac{p}{1-p}\right)}} \quad p \in [0,1].
\end{equation*}
This includes the special cases $\theta (0)=\theta(1) = 0$ and $\theta \left(\frac{1}{2}\right) = \frac{1}{2}$. 
The stochastic process $\tilde{X}_{\vc{u}}$ inherits this subgaussian behavior. More precisely, standard results such as \cite[Theorem~7.27]{Foucart2013} imply
\begin{equation*}
\mathbb{E} \left[ \mathrm{e}^{\lambda \left( \tilde{X}_{\vc{u}} - \tilde{X}_{\vc{v}} \right)} \right]
\leq \mathrm{e}^{\frac{1}{2}\theta^2 (p) \| \vc{u} - \vc{v} \|_{\ell_2}^2 \lambda^2}\quad \forall \lambda \in \mathbb{R},\; \forall \vc{u},\vc{v} \in \Sigma_s^2.
\end{equation*}
This implies that $\tilde{X}_{\vc{u}}$ is a centered subgaussian stochastic process with associated (pseudo-) metric $d(\vc{u},\vc{v}) = \theta (p) \| \vc{u}-\vc{v} \|_{\ell_2}$, see e.g.\ \cite[Definition~8.22]{Foucart2013}.
Dudley's inequality, see c.f.\ \ \cite[Theorem~8.23]{Foucart2013}, applies to such stochastic processes and yields
\begin{align*}
W_m \left( \Sigma_s^2, \tilde{\vc{a}} \right) =& \mathbb{E} \left[ \sup_{\vc{u} \in \Sigma_s^2} \tilde{X}_{\vc{u}} \right] \leq 4 \sqrt{2} \int_0^\infty \sqrt{ \log \left( \mathcal{N} \left( \Sigma_s^2, \theta (p) \| \cdot \|_{\ell_2}, u \right) \right)} \mathrm{d} u,
\end{align*}
where $\mathcal{N}\left(\Sigma_s^2, \theta (p) \| \cdot \|_{\ell_2}, u \right)$ denotes the covering number, i.e.\ the smallest integer $N$ such that there exists a subset $F$ of $\Sigma_s^2$ with $|F|\leq N$ and $\min_{\vc{y} \in F} \theta (p) \| \vc{y} - \vc{x} \|_{\ell_2} \leq u$ for all $\vc{x} \in \Sigma_s^2$. We refer to \cite[Appendix~C.2]{Foucart2013} for a concise introduction of covering numbers and their properties. In particular,
\begin{equation*}
\int_0^\infty \sqrt{ \log \left( \mathcal{N} \left( \Sigma_s^2, \theta (p) \| \cdot \|_{\ell_2}, u \right) \right)} \mathrm{d} u
= \theta (p) \int_0^1 \sqrt{ \log \left( \mathcal{N} \left( \Sigma_s^2, \| \cdot \|_{\ell_2}, v \right) \right)} \mathrm{d} v
\end{equation*}
which follows from $\mathcal{N}\left( \Sigma_s^2, \theta (p) \|\cdot \|_{\ell_2} \right) = \mathcal{N} \left( \Sigma_s^2, \| \cdot \|_{\ell_2}, \frac{u}{\theta (p)} \right)$, a change of variables in the integration ($v = \frac{u}{\theta (p)}$) and the fact that $\Sigma_s^2$ is contained in the $\ell_2$-unit ball. This last fact implies that $\mathcal{N} \left( \Sigma_s^2, \| \cdot \|_{\ell_2}, v \right) = 1$ for any $v \geq 1$ and the corresponding integrand vanishes.
For $v \in [0,1]$, the covering number of $\Sigma_s^2$ can be estimated in the following way: There are $\binom{n}{s}$ different ways to choose the support $S$ of an $s$-sparse vector in $\mathbb{R}^n$.
In turn, normalization of $\Sigma_s^2$ assures that each such vector is contained in an $s$-dimensional unit ball $B_S$. A volumetric argument in turn implies $\mathcal{N} \left( B_S, \| \cdot \|_{\ell_2}, v \right) \leq \left( 1 + \frac{2}{v} \right)^s$, see e.g.\ \cite[Prop.~C.3]{Foucart2013}.
Using subadditivity of covering numbers, we conclude
\begin{equation*}
\mathcal{N} \left( \Sigma_s^2, \| \cdot \|_{\ell_2}, v \right) \leq \binom{n}{s} \max_{|S|=s} \mathcal{N} \left( B_S, \| \cdot \|_{\ell_2}, v \right) \leq \binom{n}{s} \left( 1 + \frac{2}{v} \right)^s \leq \left( \frac{\mathrm{e}n}{s} \right)^s \left( 1 + \frac{2}{v} \right)^s,
\end{equation*}
where the last inequality is due to Stirling's formula. Combining these estimates yields
\begin{align*}
W_m \left( \Sigma_s^2, \tilde{\vc{a}} \right) \leq & 4 \sqrt{2} \theta (p) \int_0^1 \sqrt{ \log \left( \left( \frac{e n}{s} \right)^s \left( 1 + \frac{2}{v} \right)^s \right)} \mathrm{d}v \\
\leq & 4 \sqrt{2} \theta (p) \left( \sqrt{ s \log \left( \frac{\mathrm{e} n}{s} \right)} \int_0^1 \mathrm{d}v + \sqrt{s} \int_0^1 \sqrt{ \ln \left( 1 + \frac{2}{v} \right)} \mathrm{d} v \right) \\
\leq & 4 \sqrt{2} \theta (p) \left( \sqrt{ s \log \left( \frac{\mathrm{e} n}{s} \right)}+ \sqrt{s \log (3 \mathrm{e})} \right),
\end{align*}
where the last estimate follows from bounding the second integral, see e.g.\ \cite[Lemma~C.9]{Foucart2013}.

Summarizing the results from this paragraph, we conclude
\begin{align*}
W_m \left( \Sigma_s^2, \vc{a} \right) \leq & \sqrt{2s} \theta (p) \left( 4\sqrt{\log \left( \frac{e n}{s} \right)} + 4\sqrt{ \log (3 \mathrm{e})} + \frac{p}{\theta (p)}\right) 
\leq  \sqrt{2s} \theta (p) \left( 10 \sqrt{\log \left( \frac{ \mathrm{e} n}{s} \right)} + \frac{p}{\theta (p)} \right) \\
\leq & 20 \sqrt{s} \theta (p) \sqrt{ \log \left( \frac{ \mathrm{e}{n}}{s} \right) + \frac{p^2}{\theta^2(p)} },
\end{align*}
where the last line follows from $\sqrt{a}+\sqrt{b} \leq \sqrt{2 (a+b)}$ for any $a,b \geq 0$.

\subsection{Bounding $\mathrm{Pr} \left[ | \langle \vc{a},\vc{z} \rangle| \geq \theta \|\vc{z} \|_{\ell_2} \right]$ for $0/1$-Bernoulli vectors} \label{sub:Q}

In this final section we prove that for any unit vector $\vc{z} = (z_1,\ldots, z_n)^T \in \mathbb{R}^n$ ($\| \vc{z} \|_{\ell_2}=1$) and any $\theta \in [0,1/2]$, the bound
\begin{align}
 \mathrm{Pr} \left[ | \langle \vc{a}, \vc{z} \rangle | \geq  \theta \sqrt{p(1-p)}\right] 
\geq \frac{4}{13} p(1-p)(1-  \theta^2)^2
\label{eq:Qxi_appendix}
\end{align}
holds in the Bernoulli setting. Here, the probability is taken over instances of i.id.~Bernoulli vectors $\vc{a}=\sum_{i=1}^n b_i \vc{e}_i \in \mathbb{R}^n$ with parameter $p$.
Up to the multiplicative constant $\frac{4}{13}$, this bound is tight. To see this, set $\vc{z}_0=\frac{1}{\sqrt{2}}\left( \vc{e}_1-\vc{e}_2 \right)$ and observe
\begin{equation*}
\mathrm{Pr} \left[ \left| \langle \vc{a},\vc{z}_0 \rangle \right| \geq \theta\sqrt{p(1-p)} \right]
= \mathrm{Pr} \left[ \left| b_1 - b_2 \right| \geq \theta \sqrt{2 p(1-p)} \right]
= 2p (1-p)
\end{equation*}
for any $0 \leq \theta \leq \frac{1}{\sqrt{2p(1-p)}}$. Letting $\theta \to 0$ then establishes tightness.
We note in passing that a direct exploitation of the subgaussian properties of $\vc{a}$ would lead to considerably weaker results.

The derivation of Formula~\eqref{eq:Qxi_appendix} is going to rely on the Paley-Zygmund inequality and a few standard, but rather tedious, moment calculations for Bernoulli processes.
We start by exploiting
\begin{equation}
\mathrm{Pr} \left[ | \langle \vc{a},\vc{z} \rangle | \geq  \theta \sqrt{p(1-p)} \right]
= \mathrm{Pr} \left[ \langle \vc{a},\vc{z} \rangle^2\geq  \theta^2 p(1-p) \right], \label{eq:Qxi_aux1}
\end{equation}
because the latter expression is easier to handle. Introducing the nonnegative random variable
$
S := \langle \vc{a},\vc{z} \rangle^2 = \sum_{i,j=1}^n b_i b_j z_i z_j,
$, we see
\begin{equation}
\mathbb{E} \left[ S\right] = \sum_{i \neq j} \mathbb{E} \left[ b_i \right] \mathbb{E} \left[ b_j \right] z_i z_j + \sum_{i=1}^n\mathbb{E} \left[ b_i^2 \right] z_i^2 
 =  p^2 \langle \vc{1}_n,\vc{z} \rangle^2 + p(1-p) \| \vc{z} \|_{\ell_2}^2 \geq p (1-p).
\label{eq:Qxi_aux2}
\end{equation}
This calculation together with \eqref{eq:Qxi_aux1} implies
\begin{equation}
\mathrm{Pr} \left[ | \langle \vc{a},\vc{z} \rangle | \geq  \theta \sqrt{p(1-p)} \right]
\geq \mathrm{Pr} \left[ S \geq  \theta^2 \mathbb{E} \left[ S \right] \right].
\end{equation}
Since $S \geq 0$ by definition, the Paley-Zygmund inequality implies
\begin{equation}
\mathrm{Pr} \left[ S \geq \theta^2 \mathbb{E} \left[ S \right] \right]
\geq \frac{(1-\theta^2)^2 \mathbb{E} \left[ S \right]}{\mathrm{Var}(S) + \mathbb{E} \left[ S\right]^2}.
\label{eq:Qxi_aux4}
\end{equation}
We have already computed $\mathbb{E} \left[S \right]$ in \eqref{eq:Qxi_aux2}, but we still have to compute its variance. We defer this calculation to the very end of this section and for now simply state its result:
\begin{equation}
\mathrm{Var}(S)
= 2 \mathbb{E} \left[S \right]^2 - 2p^4 \langle \vc{1}_n,\vc{z} \rangle^4 
+ 4p^2 (1-p)(1-2p) \langle \vc{1},\vc{z} \rangle \sum_{i=1}^n z_i^3 
+ p(1-p) (1- 6p(1-p)) \| \vc{z} \|_{\ell_4}^4.
\label{eq:variance}
\end{equation}
We now move on to bound these contributions individually by a multiple of $\mathbb{E} \left[ S \right]^2$. 
We omit the second term and for the third term obtain
\begin{align*}
 4p^2(1-p)(1-2p) \langle \vc{1}_n,\vc{z} \rangle \sum_{i=1}^n z_i^3 
\leq &  4p^2 (1-p)^2 \langle \vc{1}_n,\vc{z} \rangle \|\vc{z} \|_{\ell_2}^3 
= 4p^2(1-p)^2 \langle \vc{1}_n,\vc{z}\rangle 
\leq 4p^2(1-p)^2 \max \left\{ \langle \vc{1}_n,\vc{z} \rangle^2, 1 \right\}  \\
\leq &  \frac{2}{p} \left( p^2 \langle \vc{1}_n,\vc{z} \rangle^2 + p(1-p)  \right)^2 
= \frac{2}{p} \mathbb{E} \left[ S \right]^2,
\end{align*}
because $\| \vc{z} \|_{\ell_2}=1$.
The fourth term can be bounded via
\begin{equation*}
 p(1-p)(1-6p(1-p)) \| \vc{z} \|_{\ell_4}^4 
\leq  p(1-p) \| \vc{z} \|_{\ell_2}^4 
\leq \frac{1}{p(1-p)} \mathbb{E} \left[ S \right]^2.
\end{equation*}
and combining all these bounds implies
\begin{align*}
\mathrm{Var}(S)
\leq  \left( 2 + \frac{2}{p} + \frac{1}{p(1-p)}\right) \mathbb{E} \left[ S \right]^2 
= \frac{3-2p^2}{p(1-p)} \mathbb{E} \left[ S \right]^2
\leq \frac{3}{p(1-p)} \mathbb{E} \left[ S \right]^2.
\end{align*}
Inserting this upper bound into the Paley-Zygmund estimate \eqref{eq:Qxi_aux4} yields
\begin{align*}
\mathrm{Pr} \left[ |\langle \vc{a},\vc{z} \rangle |\geq  \theta \sqrt{p(1-p)} \right]
\geq  \frac{(1-\theta^2)^2\mathbb{E} \left[ S\right]^2}{\mathrm{Var}(S) + \mathbb{E} \left[ S\right]^2} 
\geq  \frac{(1-\theta^2)^2 \mathbb{E} \left[ S \right]^2}{(\frac{3}{p(1-p)}+1) \mathbb{E} \left[ S\right]^2} 
 \geq  \frac{4}{13} p(1-p)(1-  \theta^2)^2,
\end{align*}
as claimed in \eqref{eq:Qxi} and \eqref{eq:Qxi_appendix}, respectively. 
In the last line, we have used $p(1-p) \leq \frac{1}{4}$ for any $p \in [0,1]$. 

Finally, we provide the derivation of Equation~\eqref{eq:variance}.
We use our knowledge of $\mathbb{E}[S] = p^2 \langle \vc{1}_n,\vc{z} \rangle^2 + p(1-p) \| \vc{z} \|_{\ell_2}^2$ together with the elementary formula
\begin{equation*}
(b_i -p)(b_j-p) = (b_i b_j - p^2) - p b_i -p b_j + 2p^2
\end{equation*}
to rewrite $S - \mathbb{E}[S]$ as
\begin{align*}
 S - \mathbb{E} \left[ S \right] 
=& \sum_{i,j=1}^n b_i b_j z_i z_j - p^2 \sum_{i \neq j} z_i z_j -p \sum_{i=1}^n z_i^2 
= \sum_{i \neq j} \left( b_i b_j -p^2 \right) z_i z_j + \sum_{i=1}^n \left( b_i^2 - p \right) z_i^2 \\
=& \sum_{i \neq j} \left( (b_i - p) (b_j -p) + p b_i + p b_j -2p^2 \right) z_i z_j + \sum_{i=1}^n \left( b_i^2 - p \right) z_i^2 \\
=& \sum_{i \neq j} \left( b_i -p \right) \left( b_j - p \right) z_i z_j + \sum_{i=1}^n \left( b_i^2 - p \right) z_i^2 
+ p \sum_{i \neq j} b_i z_i z_j + p \sum_{j \neq i} b_j z_j z_i -2p^2 \sum_{i \neq j} z_i z_j \\
=& \sum_{i \neq j} \left( b_i -p \right) \left( b_j - p \right) z_i z_j + \sum_{i=1}^n \left( b_i^2 - p \right) z_i^2 
+ 2 p \sum_{i,j=1}^n b_i z_i z_j -2p \sum_{i=1}^n b_i z_i^2 - 2p^2 \sum_{i,j=1}^n z_i z_j + 2p^2 \sum_{i=1}^n z_i^2 \\
=& \sum_{i \neq j} \left( b_i -p \right) \left( b_j - p \right) z_i z_j + \sum_{i=1}^n \left( b_i^2 - p \right) z_i^2 
+ 2p \sum_{i,j=1}^n \left( b_i - p \right) z_i z_j -2p \sum_{i=1}^n \left(b_i - p \right) z_i^2 \\
=& 2\sum_{i < j} \left( b_i -p \right) \left( b_j - p \right) z_i z_j + 2p \langle \vc{1}_n,\vc{z} \rangle \sum_{i=1}^n \left( b_i - p \right) z_i 
+ (1-2p) \sum_{i=1}^n \left( b_i - p \right) z_i^2.
\end{align*}
Here we have exploited symmetry in the first term and $b_i^2 = b_i$ to further simplify that expression.
For notational simplicity, it makes sense to re-introduce the centered random variable $\tilde{b}_i := b_i -p$:
\begin{equation*}
S - \mathbb{E} \left[ S \right]
= 2 \sum_{i < j} \tilde{b}_i \tilde{b}_j z_i z_j + 2 p \langle \vc{1}_n,\vc{z} \rangle \sum_{i=1}^n \tilde{b}_i z_i + (1-2p)\sum_{i=1}^n \tilde{b}_i z_i^2.
\end{equation*}
Employing the binomial formula $(a+b+c)^2 = a^2 +2ab + 2ac + b^2 + 2bc +c^2$, we obtain
\begin{align*}
 \mathrm{Var}(S) = & \mathbb{E} \left[ \left( S - \mathbb{E} \left[ S \right] \right)^2 \right] 
= 4 \sum_{i <j} \sum_{k <l} \mathbb{E} \left[ \tilde{b}_i \tilde{b}_j \tilde{b}_k \tilde{b}_l \right] z_i z_j z_k z_l 
+ 8 p \langle \vc{1}_n,\vc{z} \rangle \sum_{i <j} \sum_{k=1}^n \mathbb{E} \left[ \tilde{b}_i \tilde{b}_j \tilde{b}_k \right] z_i z_j z_k \\
+& 4(1-2p) \sum_{i <j} \sum_{k=1}^n \mathbb{E} \left[ \tilde{b}_i \tilde{b}_j \tilde{b}_k \right] z_i z_j z_k^2 
+ 4p^2 \langle \vc{1}_n,\vc{z} \rangle^2 \sum_{i,j=1}^n \mathbb{E} \left[ \tilde{b}_i \tilde{b}_j \right] z_i z_j \\
+& 4p(1-2p) \langle \vc{1}_n,\vc{z} \rangle \sum_{i,j=1}^n \mathbb{E} \left[\tilde{b}_i \tilde{b}_j \right] z_i z_j^2 
+ (1-2p)^2 \sum_{i,j=1}^n \mathbb{E} \left[ \tilde{b}_i \tilde{b}_j \right] z_i^2 z_j^2.
\end{align*}
Centeredness of $\tilde{b}$ together with the summation constraints ($i <j$) and $(k <l$) implies that summands in the first term vanish, unless $i = k$ and$ j=l$. This in turn implies
\begin{align*}
 4 \sum_{i <j} \sum_{k <l} \mathbb{E} \left[ \tilde{b}_i \tilde{b}_j \tilde{b}_k \tilde{b}_l \right] z_i z_j z_k z_l 
=& 4 \sum_{i <j} \mathbb{E} \left[ \tilde{b}_i^2 \right] \mathbb{E} \left[ \tilde{b}_j^2 \right] z_i^2 z_j^2 
= 2 p^2 (1-p)^2 \sum_{i \neq j} z_i^2 z_j^2 \\
=& 2p^2 (1-p)^2 \left( \sum_{i,j=1}^n z_i^2 z_j^2 - \sum_{i=1}^n z_i^4 \right)
 = 2p^2 (1-p)^2 \left( \| \vc{z} \|_{\ell_2}^4 - \| \vc{z} \|_{\ell_4}^4 \right).
\end{align*}
Using a similar argument allows us to conclude that the second and third term must identically vanish (because the index constraints $i<j$ prevents $i = j =k$ and, consequently, at least one index must always remain unpaired). 
We can exploit $\mathbb{E} \left[ \tilde{b}_i \tilde{b}_j \right] = p(1-p) \delta_{i,j}$ in the remaining terms to conclude 
\begin{align*}
\mathrm{Var}(S)
=& 2p^2 (1-p)^2 \left( \| \vc{z} \|_{\ell_2}^4 - \| \vc{z} \|_{\ell_4}^4 \right) 
+ 4 p^3 (1-p) \langle \vc{1}_n,\vc{z} \rangle^2 \| \vc{z} \|_{\ell_2}^2 \\
+& 4p^2(1-p)(1-2p) \langle \vc{1}_n,\vc{z} \rangle \sum_{i=1}^n z_i^3 
+ p(1-p)(1-2p)^2 \| \vc{z} \|_{\ell_4}^4.
\end{align*}
Slightly rewriting this expression then yields the result presented in \eqref{eq:variance}

{\footnotesize
  \bibliographystyle{IEEEtran}
  \bibliography{paper}

\begin{thebibliography}{10}
\providecommand{\url}[1]{#1}
\csname url@samestyle\endcsname
\providecommand{\newblock}{\relax}
\providecommand{\bibinfo}[2]{#2}
\providecommand{\BIBentrySTDinterwordspacing}{\spaceskip=0pt\relax}
\providecommand{\BIBentryALTinterwordstretchfactor}{4}
\providecommand{\BIBentryALTinterwordspacing}{\spaceskip=\fontdimen2\font plus
\BIBentryALTinterwordstretchfactor\fontdimen3\font minus
  \fontdimen4\font\relax}
\providecommand{\BIBforeignlanguage}[2]{{%
\expandafter\ifx\csname l@#1\endcsname\relax
\typeout{** WARNING: IEEEtran.bst: No hyphenation pattern has been}%
\typeout{** loaded for the language `#1'. Using the pattern for}%
\typeout{** the default language instead.}%
\else
\language=\csname l@#1\endcsname
\fi
#2}}
\providecommand{\BIBdecl}{\relax}
\BIBdecl

\bibitem{Vardi:1996}
Y.~Vardi, ``{Network Tomography: Estimating Source-Destination Traffic
  Intensities from Link Data},'' \emph{J. Am. Stat. Assoc.}, vol.~91, no. 433,
  pp. 365--377, 1996.

\bibitem{Castro:2004}
R.~Castro, M.~Coates, G.~Liang, R.~Nowak, and B.~Yu, ``{Network Tomography:
  Recent Developments},'' \emph{Stat. Sci.}, vol.~19, pp. 499--517, 2004.

\bibitem{Boyd2003}
J.~E. Boyd and J.~Meloche, ``{Evaluation of statistical and multiple-hypothesis
  tracking for video traffic surveillance},'' \emph{Mach. Vision Appl.},
  vol.~13, no. 5-6, pp. 344--351, 2003.

\bibitem{Donoho92}
D.~L. Donoho, I.~M. Johnstone, J.~C. Hoch, and S.~A. S, ``{Maximum Entropy and
  the Nearly Black Object},'' \emph{J. Roy. Stat. Soc. B Met.}, vol.~54, no.~1,
  1992.

\bibitem{Fuchs2005}
J.~J. Fuchs, ``{Sparsity and uniqueness for some specific under-determined
  linear systems},'' \emph{ICASSP, IEEE International Conference on Acoustics,
  Speech and Signal Processing - Proceedings}, vol.~V, pp. 729--732, 2005.

\bibitem{Zhang:2010}
G.~Zhang, S.~Jiao, X.~Xu, and L.~Wang, ``{Compressed sensing and reconstruction
  with Bernoulli matrices},'' in \emph{IEEE International Conference on
  Information and Automation}, 2010, pp. 455--460.

\bibitem{Khajehnejad2011}
M.~A. Khajehnejad, A.~G. Dimakis, W.~Xu, and B.~Hassibi, ``{Sparse recovery of
  nonnegative signals with minimal expansion},'' \emph{IEEE Trans. on Signal
  Process.}, vol.~59, no.~1, pp. 196--208, 2011.

\bibitem{Meinshausen2013}
N.~Meinshausen, ``{Sign-constrained least squares estimation for
  high-dimensional regression},'' \emph{Electron. J. Stat.}, vol.~7, no.~1, pp.
  1607--1631, 2013.

\bibitem{Slawski2013}
M.~Slawski and M.~Hein, ``{Sparse recovery by thresholded non-negative least
  squares},'' \emph{Electron. J. Stat.}, vol.~7, 2013.

\bibitem{Foucart:NNLS:2014}
S.~Foucart and D.~Koslicki, ``{Sparse Recovery by Means of Nonnegative Least
  Squares},'' \emph{IEEE Signal Proc. Let.}, vol.~21, no.~4, 2014.

\bibitem{Donoho2005}
D.~L. Donoho and J.~Tanner, ``{Sparse nonnegative solution of underdetermined
  linear equations by linear programming.}'' \emph{P. Natl. Acad. Sci. USA},
  vol. 102, no.~27, pp. 9446--9451, 2005.

\bibitem{Bruckstein2008}
A.~M. Bruckstein, M.~Elad, and M.~Zibulevsky, ``{On the uniqueness of
  non-negative sparse {\&} redundant representations},'' \emph{ICASSP, IEEE
  International Conference on Acoustics, Speech and Signal Processing -
  Proceedings}, no. 796, pp. 5145--5148, 2008.

\bibitem{Foucart2013}
S.~Foucart and H.~Rauhut, \emph{A {M}athematical {I}ntroduction to
  {C}ompressive {S}ensing}, ser. Applied and Numerical Harmonic Analysis.\hskip
  1em plus 0.5em minus 0.4em\relax Birkh\"auser/Springer, New York, 2013.

\bibitem{RudelsonZhou13}
M.~Rudelson and S.~Zhou, ``Reconstruction from anisotropic random
  measurements,'' \emph{IEEE Trans. Inform. Theory}, vol.~59, no.~6, pp.
  3434--3447, June 2013.

\bibitem{KuengGross14}
``{RIPless} compressed sensing from anisotropic measurements,'' \emph{Lin. Alg.
  Appl.}, vol. 441, pp. 110 -- 123, 2014.

\bibitem{Mendelson15}
S.~Mendelson, ``Learning without concentration,'' \emph{J. ACM}, vol.~62,
  no.~3, pp. 21:1--21:25, Jun. 2015.

\bibitem{KoltchinskiMendelson15}
V.~Koltchinskii and S.~Mendelson, ``Bounding the smallest singular value of a
  random matrix without concentration,'' \emph{Int. Math. Res. Notices}, vol.
  2015, no.~23, pp. 12\,991--13\,008, 2015.

\bibitem{Tropp2015}
J.~A. Tropp, \emph{Sampling Theory, a Renaissance: Compressive Sensing and
  Other Developments}.\hskip 1em plus 0.5em minus 0.4em\relax Cham: Springer
  International Publishing, 2015, ch. Convex Recovery of a Structured Signal
  from Independent Random Linear Measurements, pp. 67--101.

\bibitem{Dirksen16}
S.~Dirksen, G.~Lecue, and H.~Rauhut, ``On the gap between restricted isometry
  properties and sparse recovery conditions,'' \emph{IEEE Trans. Inform.
  Theory}, vol.~PP, no.~99, pp. 1--1, 2016.

\bibitem{Candes13}
E.~J. Candes, T.~Strohmer, and V.~Voroninski, ``Phaselift: Exact and stable
  signal recovery from magnitude measurements via convex programming,''
  \emph{Commun Pure Appl. Math.}, vol.~66, no.~8, pp. 1241--1274, 2013.

\bibitem{Yunyan:icassp16}
Y.~Chang, P.~Jung, C.~Zhou, and S.~Stanczak, ``Block compressed sensing based
  distributed resource allocation for m2m communications,'' in \emph{IEEE
  International Conference on Acoustics, Speech and Signal Processing
  (ICASSP)}.\hskip 1em plus 0.5em minus 0.4em\relax IEEE, 2016, pp. 3791--3795.

\bibitem{Vila2014}
J.~P. Vila and P.~Schniter, ``{An empirical-bayes approach to recovering
  linearly constrained non-negative sparse signals},'' \emph{IEEE Trans. Signal
  Process.}, vol.~62, no.~18, pp. 4689--4703, 2014.

\bibitem{Wang2011}
M.~Wang, W.~Xu, and A.~Tang, ``{A unique "nonnegative" solution to an
  underdetermined system: From vectors to matrices},'' \emph{IEEE Trans.
  Inform. Theory}, vol.~59, no.~3, pp. 1007--1016, 2011.

\bibitem{Baraniuk2008}
R.~Baraniuk, M.~Davenport, R.~A. DeVore, and M.~Wakin, ``{A Simple Proof of the
  Restricted Isometry Property for Random Matrices},'' \emph{Constr. Approx.},
  vol.~28, no.~3, pp. 253--263, jan 2008.

\bibitem{candes:rip2008}
E.~J. Candes, ``{The restricted isometry property and its implications for
  compressed sensing},'' \emph{CR Acad. Sci. I-Math}, vol. 346, no. Paris,
  Serie I, pp. 589--592, may 2008.

\bibitem{giraud:unknownnoise}
C.~Giraud, S.~Huet, N.~Verzelen \emph{et~al.}, ``High-dimensional regression
  with unknown variance,'' \emph{Stat. Sci.}, vol.~27, no.~4, pp. 500--518,
  2012.

\bibitem{sun:slasso}
T.~Sun and C.-H. Zhang, ``Scaled sparse linear regression,'' \emph{preprint
  arXiv:1104.4595}, 2011.

\bibitem{Stadler2010a}
N.~St{\"{a}}dler, P.~B{\"{u}}hlmann, and S.~van~de Geer, ``{Rejoinder:
  $\ell_1$-penalization for mixture regression models},'' \emph{Test}, vol.~19,
  pp. 209--256, 2010.

\bibitem{Kabanava2016}
M.~Kabanava, R.~Kueng, H.~Rauhut, and U.~Terstiege, ``Stable low-rank matrix
  recovery via null space properties,'' \emph{Inf. Inference}, vol.~5, no.~4,
  pp. 405--441, 2016.

\bibitem{Kabanava2015}
M.~Kabanava and H.~Rauhut, ``Analysis $\ell_1$-recovery with frames and
  gaussian measurements,'' \emph{Acta Appl. Math.}, vol. 140, no.~1, pp.
  173--195, 2015.

\bibitem{Rudelson2008}
M.~Rudelson and R.~Vershynin, ``On sparse reconstruction from fourier and
  gaussian measurements,'' \emph{Commun. Pure Appl. Math.}, vol.~61, no.~8, pp.
  1025--1045, 2008.

\bibitem{Buldygin13}
V.~Buldygin and K.~Moskvichova, ``The sub-gaussian norm of a binary random
  variable,'' \emph{Theory Probab. Math. Stat.}, vol.~86, pp. 33--49, 2013.

\bibitem{kueng:itw16}
R.~Kueng and J.~P., ``{Robust Nonnegative Sparse Recovery and 0/1-Bernoulli
  Measurements},'' in \emph{IEEE Inf. Theory Workshop (ITW)}, 2016.

\bibitem{Lawson74}
C.~L. Lawson and R.~J. Hanson, \emph{{Solving Least Squares Problems}}.\hskip
  1em plus 0.5em minus 0.4em\relax Prentice-Hall, 1974.

\bibitem{BergFriedlander:2008}
E.~van~den Berg and M.~P. Friedlander, ``Probing the pareto frontier for basis
  pursuit solutions,'' \emph{SIAM J. Sci. Comput.}, vol.~31, no.~2, pp.
  890--912, 2008.

\bibitem{Vershynin11}
R.~Vershynin, ``Introduction to the non-asymptotic analysis of random
  matrices,'' in \emph{Compressed sensing}.\hskip 1em plus 0.5em minus
  0.4em\relax Cambridge Univ. Press, Cambridge, 2012, pp. 210--268.

\end{thebibliography}
}

\end{document}